\documentclass[11pt]{article}

\usepackage{authblk}
\usepackage{fullpage} 

\usepackage[T1]{fontenc}
\usepackage{graphicx}

\usepackage{mathtools}
\usepackage{amssymb}
\usepackage{amsmath} 
\usepackage{amsthm}
\usepackage{enumitem}
\usepackage{todonotes}
\usepackage[ruled, vlined, linesnumbered]{algorithm2e}
\usepackage{float}
\usepackage{tabularx}
\usepackage{multirow}
\usepackage{array,xcolor,colortbl}
\usepackage{bm}
\usepackage{footnote}
\usepackage{subcaption}
\makesavenoteenv{tabular}
\usepackage{xargs} 
\usepackage{mathrsfs}
\usepackage{makecell}
\usepackage{xspace}
\usepackage{thmtools}
\usepackage{rotating}
\usepackage[absolute,overlay]{textpos}
\usepackage{tikz}
\usepackage[round]{natbib} 
\newcommand{\addauthor}[1]{}

\usepackage{hyperref}

\usepackage{cleveref}

\usepackage{thm-restate} 

\newtheorem{theorem}{Theorem}[section]
\newtheorem{corollary}[theorem]{Corollary}

\newtheorem{lemma}[theorem]{Lemma}

\newtheorem{observation}[theorem]{Observation}

\newtheorem{claim}[theorem]{Claim}
\crefname{claim}{claim}{claims}
\Crefname{claim}{Claim}{Claims}

\theoremstyle{definition}

\newtheorem{example}[theorem]{Example}

\newenvironment{claimproof}
{
    
    \proof
}
{
    \endproof
    
}

\newcommand{\proofqed}{} 
\newcommand{\claimqed}{}

\usepackage{color}

\newcommand{\complexityclass}[1]{\mathcal{#1}}
\newcommand{\p}{\ensuremath{\complexityclass{P}}}
\newcommand{\np}{\ensuremath{\complexityclass{NP}}}
\newcommand{\bigO}[1]{\ensuremath{\mathcal{O}\left(#1\right)}}

\newcommand{\ega}{(\textsc{EG})} % K^2: Maybe this way? With brackets? So it's also smaller for equations
\newcommand{\uti}{(\textsc{UT})}
\newcommand{\egsup}{^{\ensuremath{\scriptscriptstyle \ega}}}
\newcommand{\utsup}{^{\ensuremath{\scriptscriptstyle \uti}}}
\newcommand{\Deg}{D\egsup} %DP with egal cost
\newcommand{\Dut}{D\utsup} %DP with util cost
\newcommand{\Dstar}{D^{*}} %DP with util cost
\newcommand{\Meg}{M\egsup} %DP with egal cost
\newcommand{\Mut}{M\utsup} %DP with util cost
\newcommand{\Mstar}{M^{*}} %DP with util cost

\newcommand{\gbp}{transit investment problem} %was previously "Graph-Based Bus Stop Problem"
\newcommand{\rwp}{railway network design problem} %was previously "Graph-Based Bus Stop Problem"
\newcommand{\RWP}{Railway Network Design Problem} %was previously "Graph-Based Bus Stop Problem"

\newcommand{\gbpins}{\langle G, \ac, \alpha, \beta \rangle}

\newcommand{\problemclass}[1]{\texttt{#1}}

\newcommand{\GraphText}{TIP}
\newcommand{\Graph}{\problemclass{\GraphText}}
\newcommand{\GraphDec}{\problemclass{{\GraphText}\ensuremath{^{\hspace{1pt}\star}}-Dec}}
\newcommand{\GraphHub}{\problemclass{TIPH}}
\newcommand{\GraphMax}{\ensuremath{\problemclass{\GraphText}^\problemclass{\hspace{1pt}\ensuremath{{\ega}}}}}
\newcommand{\GraphSum}{\ensuremath{\problemclass{\GraphText}^\problemclass{\hspace{1pt}\ensuremath{{\uti}}}}}
\newcommand{\GraphStar}{\ensuremath{\problemclass{\GraphText}^{\hspace{1pt}\star}}}

\newcommand{\GraphHubStar}{\ensuremath{\problemclass{\GraphHub}^{\hspace{1pt}\star}}}

\newcommand{\Sat}{3-\problemclass{SAT}}

\newcommand{\SetCover}{\problemclass{SetCover}}
\newcommand{\MCClique}{\problemclass{MulticoloredClique}}
\newcommand{\Rail}{\problemclass{RDP}}
\newcommand{\RailSum}{\ensuremath{\problemclass{RBP}^\problemclass{\hspace{1pt}\ensuremath{{\uti}}}}}
\newcommand{\RailMax}{\ensuremath{\problemclass{RBP}^\problemclass{\hspace{1pt}\ensuremath{{\ega}}}}}
\newcommand{\RailStar}{\ensuremath{\problemclass{RDP}^{\hspace{1pt}\star}}}

\newcommand{\gdyup}{{\mathscr G_{\uparrow}}}
\newcommand{\gdydn}{{\mathscr G_{\downarrow}}}

\DeclarePairedDelimiter{\multiset}{\{\!\{}{\}\!\}}

\newcommand{\N}{\mathbb{N}}
\newcommand{\Nz}{\mathbb{N}_0}

\newcommand{\ac}{\mathcal{A}}

\newcommand{\vc}{E}

\newcommand{\ic}{\ensuremath{\mathcal{I}}}

\newcommand{\bic}{\ensuremath{\mathcal{H}}}

\newcommand{\algc}{\mathrm{ALG}}
\newcommand{\optc}{\mathrm{OPT}}
\newcommand{\deltag}{\delta_{\mathrm{guess}}}

\newcommand{\opt}{\text{opt}}

\title{Fair and Efficient Investment in Public Transportation}

\author[1]{Martin Bullinger}
\author[2]{Edith Elkind}
\author[3]{Kassian Köck}

\affil[1]{ \small School of Engineering Mathematics and Technology, University of Bristol, Bristol, UK}
\affil[2]{ \small School of Engineering, Northwestern University, Evanston, USA}
\affil[3]{ \small School of Computation, Information and Technology, Technical University of Munich, Munich, Germany\protect\\ \vspace*{0.05cm} martin.bullinger@bristol.ac.uk, edith.elkind@northwestern.edu, kassian.koeck@tum.de}

\date{}

\begin{document}
\maketitle              
\begin{abstract}
We study a stylized model of infrastructure investment in public transportation.
In our model, each agent travels between a pair of terminals in a network captured by a weighted graph, where edge weights represent distances.
The central planner can reduce the travel time along a fixed number of edges, with the goal of maximizing the utilitarian or egalitarian welfare.
When there is only one agent, we provide a polynomial-time algorithm that combines Dijkstra's algorithm with a dynamic program.
We then demonstrate how to use this algorithm as a subroutine to solve the problem for two agents.
Generalizing this idea, we present an XP algorithm parameterized by the number of agents $n$; however, our problem turns out to be W[1]-hard with respect to $n$. Nevertheless, we establish a fixed-parameter tractability result for the special case where all agents travel to a common hub.
If the number of agents is variable, we obtain \np-completeness and inapproximability results.
We discuss implications of our results for a related model of railway network design.\smallskip

\end{abstract}
	
\section{Introduction}

Urban transport is a key element of modern economy: it  contributes to the overall societal welfare 
by providing access to wealth-generating activities \citep{banister2006city}. 
As cities continue to grow and urban areas become increasingly denser, urban planners face multiple challenges,  
such as congestion, environmental impact, and unequal access to services \citep{banister2008sustainable}. Accordingly, there is a significant body of research focused on optimizing transit systems \citep{nnene2023optimising}. 

Traditionally, the design and optimization of public transportation systems have been driven primarily by efficiency objectives, such as minimizing total travel time, travel cost, or overall system operating costs. 
These objectives underpin much of the literature on transit network design and remain central performance measures in transportation planning \citep{dDWi24a,SVL24a}.

However, recent studies  have argued that, by focusing on efficiency, urban planners may overlook important aspects of service quality, such as equitable distribution of transit resources among diverse population groups \citep{litman2017evaluating}. 
Accordingly, an important research challenge is to develop transit planning frameworks that balance efficiency with social equity, so that improvements in operational performance do not come at the expense of marginalized communities \citep{sanchez2003moving}.

This tradeoff can be phrased in terms of utilitarian vs.~egalitarian social welfare.
The former aims to maximize the sum of agents' utilities.
In the context of public transportation, this means reducing total travel times, lowering operational costs, and increasing service frequency to optimize mobility for the greatest number of users \citep{vickrey1969congestion}.
However, while utilitarian principles can lead to highly efficient systems, they may also inadvertently reinforce inequalities, as they tend to favor high-demand areas and neglect regions with lower population densities or reduced economic activity \citep{amaral2009auction}.

In contrast, egalitarian welfare prioritizes fairness and equitable distribution of resources, ensuring that no individual or group is disproportionately disadvantaged \citep{Rawl71a}.
In transportation planning, this means designing systems that provide access not only for those in high-density, economically vibrant areas but also for individuals in lower-income or remote communities \citep{behbahani2019conceptual}.
Equity-based transit planning ensures an inclusive urban environment where mobility is a right accessible to all, rather than a privilege enjoyed primarily by those in high-density regions \citep{behbahani2019conceptual}.

Against this background, we study a stylized model of investment in public transportation, which we term the \emph{\gbp{}}. 
In our model, a set of $n$ agents 
travel on a network (modeled as a weighted graph): each agent is associated with a pair of terminals and travels along a shortest path between these terminals. 
A central planner can make a limited number of infrastructure improvements, by choosing $\beta\in\mathbb N$ edges and reducing the travel time along each of the selected edges by a discount factor $\alpha\in [0, 1]$.
The aim is to maximize utilitarian or egalitarian welfare.
The infrastructure improvements captured by this model include, e.g., closing a residential road to motorized traffic to make bicycle journeys faster, or introducing an express bus route between two key terminals that skips the intermediate stops. While our model is quite simple, in that all edges have the same investment cost, and the improvement in travel time is captured by a single multiplicative parameter $\alpha$, our results demonstrate that it offers interesting algorithmic challenges;
understanding the complexity of determining the optimal investment policy for this simple model is a prerequisite to studying richer, more expressive models.

We start by exploring two natural greedy algorithms for the planner's problem, and show that they do not yield nontrivial approximation guarantees.

Then we focus on the case $n=1$, 
and show that, for a single agent, optimal investment can be computed in polynomial time.
Our algorithm can be viewed as a dynamic-programming 
variant of Dijkstra's algorithm, where we iteratively identify optimal travel distances to a given destination with a fixed budget.

Subsequently, we extend this result, first to $n=2$ and then to any fixed number of agents, obtaining an XP algorithm with respect to $n$. Our approach works for  both the utilitarian and  the egalitarian version of our problem.
For two agents, the idea is to guess the intersection points of the two paths and then to apply the single-agent algorithm for each segment. Extending this approach beyond two agents 
is nontrivial, as the path-intersection dependencies of shortest path systems can be quite complex. To address this issue, we exploit the existence of special shortest path systems, where each pair of shortest paths only intersects on one contiguous sub-path.

A natural follow-up question is whether our XP algorithm can be strengthened to an FPT algorithm.
We derive a W[1]-hardness result for the general case, but provide an FPT algorithm under the natural domain restriction where all agents travel to a common hub.
While this algorithm is quite different from our XP algorithm, it once again relies on a structural property of shortest path systems: we show that when all agents share a hub, there is a shortest path system that forms a tree.
This observation enables us to formulate a dynamic program that relies on tracing non-leaves of such a tree as well as agents' individual (i.e., non-hub) terminals.
To improve the running time, we intertwine this approach with a second---concurrently executed---dynamic program.

While our W[1]-hardness result 
implies that our problem is hard when $n$ is not fixed, it does not rule out the existence of good approximation algorithms. To conclude the paper, we argue that, for general $n$, optimal utilitarian/egalitarian welfare is hard to approximate, too. 
Specifically for $\alpha = 0$, i.e., if an investment into an edge reduces the travel time to~$0$, we show that the existence of a polynomial-time algorithm with a finite approximation guarantee for utilitarian or egalitarian welfare implies \p = \np. 
Further, if $\alpha\in (0,1)$, no polynomial-time algorithm can approximate egalitarian welfare by a factor of
less than $\frac{\alpha + 1}{2\alpha}$ unless \p = \np.
Note that this expression is unbounded as $\alpha$ approaches $0$ and decays to~$1$ when $\alpha$ converges to~$1$.
Hence, a decision maker that cares about egalitarian welfare faces substantial computational challenges.

\section{Related Work}

Public transport design is often modeled as a multi-objective vehicle routing problem (VRP) \citep{jozefowiez2008multi}. 
For example, \addauthor{Jozefowiez et al.~}\citet{jozefowiez2009evolutionary} introduce a bi-objective VRP framework that aims to minimize both the total route length and the imbalance across routes, defined as the difference between the longest and shortest route lengths.

In addition to VRPs, substantial research has focused on optimizing existing transportation schedules to balance operator and passenger costs. 
For instance, \addauthor{Yu et al.~}\citet{yu_parallel_2011} propose a formal model to determine optimal bus headways, factoring in constraints such as passenger demand, vehicle capacity, and operational resources, to minimize waiting, boarding, and operating costs. 
Building upon this, \addauthor{Huang et al.~}\citet{huang_customized_2020} develop a real-time dynamic route optimization framework for customized bus services. 
Their approach trades off customer satisfaction and operator profitability by allowing flexible route adjustments to accommodate real-time requests.  
\addauthor{Yan et al.~}\citet{yan_robust_2012}, on the other hand, target schedule stability by introducing optimal slack times into bus schedules to minimize deviations.

While much of the work in this domain considers general transit challenges, some research focuses on specific applications, such as the school bus routing problem, where students need transportation from scattered pick-up points to a common school location. 
\addauthor{Konstantinos and Dimitra~}\citet{konstantinos_school_2023} employ Dijkstra's algorithm to optimize route starting points, reducing total student travel time, walking distances, and bus fleet size.
However, in contrast to our setting, their algorithm does not allocate funding to improve travel times.
Similarly, \addauthor{Arias-Rojas et al.~}\citet{arias-rojas_solving_2012} apply an ant colony optimization metaheuristic to achieve efficient routing under complex constraints.
We note that our problem (and in particular our algorithm for $n=2$) is similar to the disjoint-paths problem; however, the latter is known to be fixed-parameter tractable with respect to the number of terminal pairs \citep{RoSe1995a}, while our problem is W[1]-hard with respect to $n$. 

Recent work considers issues of fairness in public transportation network design.
Most closely related to our work, \addauthor{He et al.~}\citet{HBL+24} define a model for railway network design.
They primarily focus on an empirical analysis of railway systems in nine countries, and
consider a spectrum of welfare notions interpolating between utilitarian and egalitarian welfare. 
In particular, they derive a hardness result for egalitarian welfare in the restricted setting where taking certain routes has infinite cost.
In \Cref{app:railway}, we discuss the relationship between their model and ours, and show how some of our hardness and inapproximability results for fixed cost parameters extend to their setting \citep{BEK26a}.

The problem of optimally upgrading a limited number of edges in a network has been considered for other network optimization problems.
\addauthor{Campbell et al.~}\citet{campbell2006upgrading} aim to minimize worst-case travel times in a network, as measured by its diameter; in contrast to our demand-driven welfare objectives, this leads to an all-pairs minimax guarantee. 
\addauthor{Duque et al.~}\citet{duque2013accessibility} model network transit as a budget-constrained minimum cost flow problem.
Vertices can demand or supply flow, and the goal is to obtain a minimum cost flow by upgrading edges subject to a budget constraint; edges can have several upgrade levels at different costs.
The authors provide an integer program that captures this problem, prove an \np-hardness result, and then introduce and analyze two heuristics.
\addauthor{Landete et al.~}\citet{landete2023upgrading} study budgeted edge upgrading for the graphical traveling salesman problem. 
Similar to \addauthor{Duque et al.~}\citet{duque2013accessibility}, they allow for multiple upgrade levels with different costs and benefits.
A key difference with our model is that these works optimize the entire network, without considering multiple agents and their travel needs. 
The closest to our work is the contribution by \addauthor{Lin and Mouratidis~}\citet{lin2015best}, who study a similar model and aim to maximize the utilitarian welfare under budget-constrained edge upgrades. 
They provide practical heuristics based on the iterative consideration of paths and test them experimentally.
Our results complement their work by formally investigating the power and limitations of polynomial-time algorithms.

Finally, a recent line of research explores investing in vertices rather than edges of a network, modelling, e.g., access points to a bus network.
\addauthor{Bullinger et al.~}\citet{BEL25a} explore the placement of bus stops along a linear route, introducing utilitarian fairness as a key metric alongside total travel time. 
Their approach is rooted in facility location theory \citep{ChanFLLW21} and proportional fairness in computational social choice \citep{LaSk22b}.
\addauthor{Aziz et al.~}\citet{AGGV26} extend their analysis to more general transit networks and establish connections to fair clustering.
A special case of their model, where $\alpha = 0$ and only two bus stops are built, is considered by \addauthor{Chan and Wang~}\citet{ChWa23a}.
Moreover, the same cost function is also applied by \addauthor{Fukui et al.~}\citet{FSN20a}.
However, both works focus on strategic aspects rather than welfare.

\section{Preliminaries}

For a positive integer $i$, we write $[i] := \{1,\dots, i\}$ and  $[i]_0 := \{0,\dots, i\}$. We represent multi-sets by $\multiset{\cdot}$  and multi-set inclusion by $\sqsubseteq$.

Consider a weighted graph $G = (V,E,w)$, where $w : E \rightarrow \mathbb{Q}_0^+$.
For each $u\in V$, we denote by $N(u) :=\{v\in V\colon \{u,v\}\in E\}$ the set of \emph{neighbors} of $u$.
A \emph{path} $P = (V',E')$ of length $\ell$ in $G$ is a subgraph of $G$ such that $V' = \{v_0,v_1,\dots, v_{\ell}\}$ with $v_i\neq v_j$ for $0\le i < j \le \ell$ and $E' = \{\{v_{i-1},v_i\}\colon i\in [\ell]\}\subseteq E$.
Given a path $P$, we refer by $V(P)$ and $E(P)$ to its vertices and edges.
For $u,v\in V$,  
we denote by $P(u,v)$ the set of all paths connecting $u$ and $v$.
For $F\subseteq E$, the {\em weight} of $F$ is $w(F) := \sum_{e\in F} w(e)$. We extend this notation to paths: for a path $P$, we write $w(P) := w(E(P))$.

An instance $\bic = \gbpins$ of the \emph{\gbp{}} (\Graph{}) is given by a weighted graph $G = (V,E,w)$, where $w : E \rightarrow \mathbb{Q}_0^+$, 
a multi-set of $n$ \emph{agents} $\ac \sqsubseteq V^2$, a \emph{discount factor} $\alpha \in [0,1)$, and a \emph{budget} $\beta \in \N$. 
A \emph{solution} is a subset of edges $S \subseteq E$; it is considered \emph{feasible} if $|S| \leq \beta$. 
We assume without loss of generality that 
$G$ is connected; otherwise, its disconnected subgraphs can be treated separately.
We also consider the restriction of \Graph{} to instances where all agents share one terminal, i.e, instances $\bic = \gbpins$ with $G = (V,E,w)$ such that $\ac \sqsubseteq V\times \{t\}$ for some $t\in V$.
We refer to the shared terminal $t$ as a \emph{hub} and to \Graph{} restricted to instances with a hub as \GraphHub.
Moreover, 
given a $k\in \mathbb N$ and $\alpha \in [0,1)$, we denote by $k$-\Graph{} and $\alpha$-\Graph{} instances of \Graph{} with exactly $k$ agents and with discount factor fixed to $\alpha$, respectively. 

For an agent $a = (s, t) \in \ac$, the points $s$ and $t$ are called the \emph{terminals} of $a$; the \emph{set of terminals} is $\theta(\ac) := \{s,t : (s,t) \in \ac\}$. 
The \emph{base cost} of agent $a$ is 
$d(a) := \min_{P \in P(u,v)}w(P)$. 
Given a solution $S\subseteq E$,
the \emph{reduced cost} of $a$ for $S$ is defined as 
$\delta_S(a) := \min_{P \in P(u,v)}w_{S}(P)$ where 
$w_{S}(e) := \alpha\cdot w(e)$ if $e \in S$ and $w_S(e) = w(e)$ otherwise.
The \emph{egalitarian} (respectively, \emph{utilitarian}) cost of $S$ is defined as $c\egsup(S) := \max_{a \in \ac}\delta_S({a})$ (respectively, $c\utsup(S) := \sum_{a \in \ac}\delta_S({a})$). 
We use the superscript $\star$ to indicate analogous treatment of both $\ega$ and $\uti$, e.g., $c^{\star} = c\egsup / c\utsup$.

We consider the following optimization problem, which we denote by \GraphStar: given an instance $\bic$ of \Graph, find a feasible set of 
edges $S$ that minimizes $c^\star(S)$.
The associated decision problem, denoted by \GraphDec, 
takes an instance $\bic$ of \Graph{} and a rational number $\kappa$ as an input,  and asks if $\bic$ admits a feasible solution $S$ with $c^{\star}(S)\le \kappa$.
We denote by \GraphHubStar{} (respectively, $k$-\GraphStar{}, $\alpha$-\GraphStar{}) the restrictions of \GraphStar{} to instances in \GraphHub{} (respectively,  $k$-\Graph{}, $\alpha$-\Graph{}).

Given an instance $\bic$ of \Graph, 
let $\kappa_{\opt}^{\star}(\bic) := \min_{S \subseteq \vc, |S| \leq \beta} c^{\star}(S)$. 
For a $\gamma \ge 1$,
an algorithm $S$ is called a $\gamma$-\emph{approximation algorithm} for \GraphStar{} if, given an instance of \Graph{} $\bic$, it produces a feasible solution $S(\bic)$ with cost $c^{\star}(S(\bic)) \leq \gamma \cdot \kappa_{\opt}^{\star}(\bic)$.
If \GraphStar{} admits a polynomial-time $\gamma$-approximation algorithm, it is said to be $\gamma$-\emph{approximable}.
Further, \GraphStar{} is said to be \emph{$\gamma$-inapproximable} if no polynomial-time $\gamma$-approximation algorithm exists, unless $\p = \np$;
\GraphStar{} is said to be \emph{inapproximable} if it is $\gamma$-inapproximable for every $\gamma\ge 1$.

The following example illustrates the key features of our model.

\begin{example}\label{ex:intro} 
Consider an instance $\gbpins$ of \Graph{} with the underlying graph given in \Cref{fig:graph:1}. 
Two agents $a_1=(s,t)$ and $a_2=(s',t')$ traverse the graph. 
We consider $\alpha=0$, i.e., discounted edges incur no cost, and a budget of $\beta=1$. 
If no edges are discounted, then each agent’s base cost equals the weight of their shortest path, i.e., $d(a_1)=10$ and $d(a_2)=13$. 
A simple heuristic to select a single edge to be discounted is to pick a maximum-weight edge on some agent’s shortest path (in this case, the orange edge of weight $6$).
This yields $\delta_S(a_1) = \delta_S(a_2) = 7$.
Note that $a_1$ can also use the discounted edge to improve her travel cost.
However, the optimal utilitarian and egalitarian cost is attained by discounting the dashed edge of weight~$10$ marked in the figure.
This yields travel costs of $\delta_S(a_1) = 4$ and $\delta_S(a_2) = 5$. 
This example illustrates how a modification outside of each shortest path can help all agents at the same time.\hfill$\lhd$
\end{example}

\begin{figure}[tb]
\centering
\includegraphics[width=.9\textwidth, page=3, trim=30 320 540 70, clip]{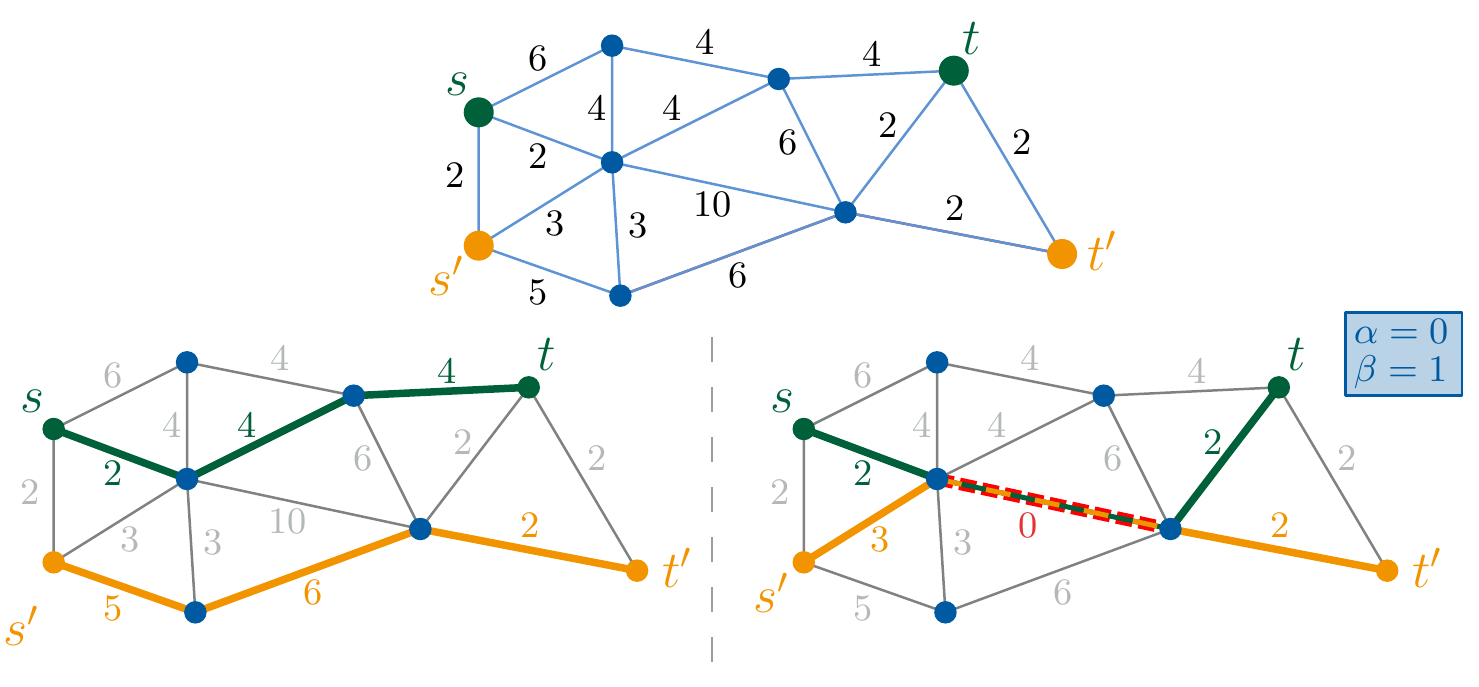}
\caption{Illustration of the \gbp.
The upper panel shows the original graph, the left panel depicts the shortest paths used under base cost, and the right panel demonstrates how reducing the cost on the dashed edge alters these routes. 
\label{fig:graph:1}}
\end{figure}

We conclude this section with some simple observations.
First, we note that an agent’s travel cost for a given solution can be evaluated in polynomial time by using a shortest-path algorithm on the graph where the weight of the selected edges is discounted by a factor of $\alpha$. 
This implies that \GraphDec{} is in \np{} for both cost functions.

\begin{observation}\label{obs:NPmemb}
    \GraphDec{} is contained in \np.
\end{observation}

Second, we derive a simple approximation guarantee.
If $0<\alpha< 1$ 
and our budget is unlimited, we can scale down the cost of every edge by a factor of $\alpha$ to obtain a solution of cost $\alpha\cdot c^\star (\emptyset)$. Hence, $c^\star(E) = \alpha\cdot c^\star (\emptyset)$, and therefore the trivial solution $S=\emptyset$ is an $\alpha^{-1}$‑approximation in every instance.

\begin{restatable}{observation}{GBPApprox} 
\label{obs:TrivApprox}
\GraphStar{} is $\alpha^{-1}$-approximable for $\alpha \in {} (0,1)$.
\end{restatable}

\section{Results}
In this section, we present our results. 
We begin by considering two natural greedy algorithms and show that they do not provide any approximation guarantees.
Towards algorithmic feasibility, we then introduce a modified version of Dijkstra's algorithm that yields optimal solutions to \GraphStar{} for a single agent. 
Subsequently, we analyze in depth instances with a fixed number of agents.
We conclude the section with \np-hardness and inapproximability results for a variable number of agents. 

\subsection{Greedy Algorithms}\label{sec:gdy}

A natural strategy to solve $\GraphStar{}$ is to build solutions incrementally via locally optimal improvements. The \emph{bottom-up greedy algorithm} $\gdyup$ iteratively selects an edge that leads to the best improvement in (utilitarian or egalitarian) cost and adds it to the current set of selected edges.
Similarly, the \emph{top-down greedy algorithm} $\gdydn$ starts with the full set of edges and iteratively removes the edge whose removal results in the smallest cost, until feasibility is established.
Formally, the \emph{bottom-up greedy algorithm} $\gdyup$ constructs a solution by adding edges iteratively:
\begin{enumerate}
    \item Set $S_0 = \emptyset$.
    \item For all $i \in [\beta]$: Let $e\in \arg\min_{e \in E \setminus S_{i-1}}c^\star(S_{i-1} \cup \{e\})$ and set $S_i = S_{i-1} \cup \{e\}$.
\end{enumerate}
As the final output, $\gdyup$ returns $S_\beta$.

The \emph{top-down greedy algorithm} $\gdydn$ for a graph with $\ell$ edges works in reverse by deleting edges iteratively:
\begin{enumerate}
    \item Set $S_{\ell} = E$
    \item For all $i = \ell - 1,\dots, \beta$: Let $e\in \arg\min_{e \in S_{i+1}} c^\star(S_{i+1} \setminus \{e\})$ and set $S_i = S_{i+1}\setminus \{e\}$.
\end{enumerate}
As output, $\gdydn$ returns $S_\beta$.

Unfortunately, both algorithms perform poorly,  both for the egalitarian and for the utilitarian objective.
Specifically, they fail to achieve more than an $\alpha^{-1}$-approximation, i.e., the trivial approximation guarantee from \Cref{obs:TrivApprox}.
The proof idea is to create an instance where each agent has an individual direct route that is slightly shorter than a common ``motorway.'' 
The greedy algorithms will select edges on the individual routes, while a much better solution is achieved by choosing all edges on the motorway.
We formally prove the result in \Cref{count:greedy}.
All other omitted proofs are also contained in the appendix.

\begin{restatable}{proposition}{gdpperformance}\label{prop:gdy}
    Let $\alpha \in (0,1)$ and $1\le \gamma < \alpha^{-1}$, or $\alpha = 0$ and $\gamma \ge 1$.
    Then $\gdyup$ and $\gdydn$ are not $\gamma$-approximation algorithms for $\alpha$-$\GraphStar{}$. 
\end{restatable}

\subsection{Single-Agent Case}

For a single agent, the egalitarian and utilitarian objectives are identical, but investing the budget optimally is still challenging.
Indeed, \Cref{ex:intro} illustrates that one cannot simply invest in the most expensive edges on an agent's shortest path.
Nevertheless, we now present \Cref{alg:dijkstra_mod}: a modified version of Dijkstra's algorithm that computes optimal solutions in polynomial time.

\begin{restatable}{theorem}{OneG}
\label{theo:1-graph}
Given a graph with $m$ vertices and $\ell$ edges, 
\Cref{alg:dijkstra_mod} solves $1$-\GraphStar{} in time \bigO{\beta m\log m + \beta \ell}.
Moreover, within this time bound, given a fixed vertex $s$, \Cref{alg:dijkstra_mod} computes
the minimum cost of a path from $s$ to $v$ when reducing the cost of at most $b$ edges, 
for all pairs $(v,b)$ with $v\neq s$ and $b\in [\beta]_0$.
\end{restatable}

A formal proof of correctness and running time analysis are given in \Cref{app:oneag}.
In that appendix, we also provide an example execution of the algorithm.
In what follows, we focus on describing the algorithm and giving an intuition for its running time.

We refer to a pair $(v, b)\in V \times [\beta]_0$, as a \emph{routing pair}.
\Cref{alg:dijkstra_mod} operates similarly to Dijkstra's algorithm, but instead of 
vertices, it processes routing pairs.
The idea is to compute, for each routing pair $(v,b)$, the length of a shortest path from $s$ to $v$ when discounting at most $b$ edges.
These distances are captured by variables $D[v, b]$ that are iteratively determined by the algorithm. 
We initialize $D[v, b] = \infty$ for each routing pair $(v, b)$ with the exception of the starting pair $(s, 0)$, for which we set $D[s, 0] = 0$.

In each iteration, we select a 
routing pair $(v,b)$ with the smallest value $D[v,b]$ among unselected pairs, and refer to it as the \emph{pivot pair}.
The minimality ensures that it would not be selected in further updates.
We use the pivot pair to update distances to ``adjacent'' routing pairs. 
There are three types of updates: 
\begin{itemize}
    \item \textbf{Non-reducing update}: simulates a path where the final edge cost is not reduced.
    \item \textbf{Reducing update}: simulates a path to a neighbor where the final edge cost is reduced.
    \item \textbf{Budget-increasing update}: simulates a path that is using less than the maximum budget.
\end{itemize}

The first two update routing pairs of neighbors of~$v$ via paths ending with an edge from~$v$; the third uses a known shortest path even under a higher budget.
We could, in principle, stop the algorithm once we have processed $(t,\beta)$.
However, in our presentation, we explore all possible pairs, thereby finding paths from $s$ to all other vertices and under any possible budget up to $\beta$. 

\begin{algorithm}[tb!]
    \caption{Modified Dijkstra's Algorithm for $1$-\GraphStar{}     \label{alg:dijkstra_mod}}
      \DontPrintSemicolon 
      \SetFuncSty{textsc}
      \SetKwFor{ForAll}{forall}{do}

        \SetDataSty{textnormal}
      \SetKwData{Q}{$Q$}
      \SetKwData{D}{$D$}
      \SetKwData{p}{$p$}
      \SetKwData{S}{$S$}

    \KwIn{Instance $\gbpins$ of \Graph{} with $|\ac| = 1$}
    \KwOut{Distance matrix \D with the minimal cost for all routing pairs, optimal solution $S$.}

    \tcp{Initialization}
    \Q $\gets \{(v,b)\colon v\in V\setminus \{s\}, b \in [\beta]_0\}\cup\{(s, 0)\}$ \tcp{Priority queue}
    $\D[v, b] \gets \infty$, $\p[v,b]\gets \bot$ for all $v \in V \setminus \{s\}, b \in [\beta]_0$ \;
    $\D[s, 0] \gets 0$, $\p[s, 0] \gets \bot$ \;
    \tcp{Modified Dijkstra Algorithm}
    \While{\upshape $\Q\neq \emptyset$}{
        \tcp{Extract routing pair with minimal distance}
        Select $(v, b)\in \Q$ minimizing $\D[v,b]$ \label{ln:Qextract}\;
		$\Q \gets \Q\setminus \{(v,b)\}$ \;
        \ForAll{$u \in N(v)\setminus \{s\}$}{
            \tcp{Non-reducing update}
            \If{$\D[v, b] + w(v, u) < $ $\D[u, b]$}{
            $\D[u, b] \gets \D[v, b] + w(v, u)$ and $\p[u, b] \gets (v,b)$\label{ln:update_nonred}\;
            }
            \tcp{Reducing update}
            \If{$(b < \beta) \land (\D[v, b] + \alpha \cdot w(v, u) < $ $\D[u, b + 1])$}{
                $\D[u, b + 1] \gets \D[v, b] + \alpha \cdot w(v, u)$ and $\p[u, b + 1] \gets (v,b)$ \label{ln:update_red}\;
            }
        }
        \tcp{Budget-increasing update}
        \If{$(v \neq s) \land (b < \beta) \land (\D[v, b] < \D[v, b+1])$}{
            $\D[v, b+1] \gets \D[v, b]$ and $\p[v, b + 1] \gets (v,b)$ \label{ln:update_budget}\;
        }
    }
    \tcp{Retrieval of solution}
    $\S\gets \emptyset$, $r\gets (t,\beta)$\;
    \While {$\p[r]\neq \bot$}
    {Let $r = (v_1,b_1)$ and $\p[r] = (v_2,b_2)$\;
    \If{$v_1\neq v_2$ and $b_1 = b_2 + 1$}{
    $\S\gets \S\cup \{v_1,v_2\}$\;
    }
    $r\gets \p[r]$\;
    }
    
    \Return $\D[t, \beta]$, $\S$ \;
\end{algorithm}

The necessity of the budget-increasing update may not be immediately obvious. Its purpose is to enforce $D[v,b] \leq D[v,b-1]$ for all routing pairs with budget indicator $b\ge 1$.
If this update is omitted, the algorithm would determine the minimum cost with budget \textit{exactly} $\beta$, rather than \textit{at most} $\beta$.
This would require a final check going through $D[t,b]$ for all $b\in [\beta]_0$.

\Cref{alg:dijkstra_mod} can be viewed as a modification of Dijkstra's algorithm. 
Instead of operating on a graph with $m$ vertices and $\ell$ edges, it is based on $\beta \cdot m$ routing pairs and $\beta \cdot \ell$ edges. 
While we defer the proof of correctness, we now present some intuition regarding its running time.
The overall running time depends, similarly to that of Dijkstra's algorithm, on the efficiency of the data structure $Q$. 
Let $T_\text{dk}$ and $T_\text{em}$ represent the time costs of the \emph{decrease key} and \emph{extract minimum} operations, respectively. 
The former is responsible for updating the significance of a routing pair in the priority queue in lines~\ref{ln:update_nonred}, \ref{ln:update_red}, and \ref{ln:update_budget}.
The latter is relevant for extracting pivot pairs in \cref{ln:Qextract}.

The total cost is therefore:
$$\bigO{\beta m \cdot T_\text{em} + \beta \ell \cdot T_\text{dk}}$$
as demonstrated by \addauthor{Fredman and Tarjan~}\citet{ft-fhtui-87}. 
With a Fibonacci heap~\citep{thomas2009introduction} ($T_{\mathrm{em}}=\bigO{\log(\beta m)}$, $T_{\mathrm{dk}}=\bigO{1}$ amortized) this becomes 
\[
\bigO{\beta m\log m + \beta\ell}\text. 
\]

Note that, in any connected graph, the number of edges $\ell$ satisfies $m-1 \leq \ell \leq m^2$. 
Hence, the runtime can be written more compactly as $\bigO{\beta m^2}$.
In addition, in principle, Dijkstra's algorithm computes the shortest paths from a source to all vertices in the graph. 
Similarly, \Cref{alg:dijkstra_mod} computes the minimum distance from a source vertex $s$ to all other vertices using budgets from $0$ to $\beta$ in a single execution, which we refer to as a \emph{pivot pair mapping}, i.e., a mapping $V\setminus \{s\} \times [\beta]_0\to \mathbb Q_0^+$.
The restriction of such a mapping to a fixed target vertex is called a
\emph{budget mapping}, i.e, these are mappings $[\beta]_0\to \mathbb Q_0^+$.
Note that \Cref{theo:1-graph} establishes pivot pair mappings for a fixed source vertex. Thus, we obtain the following corollary,
which will be useful for extending the algorithm to settings with more than one agent.

\begin{corollary}\label{cor:budmap}
    Given an instance of \Graph{}, we can compute budget mappings for all pairs of terminals in time $\mathcal O(\beta m^2\log m + m\beta \ell)$.
\end{corollary}

\subsection{Fixed Number of Agents}
A natural question is whether \Cref{theo:1-graph} extends to instances with more than one agent. 
Consider first the case of two agents.
In an optimal solution, the agents’ shortest paths may be disjoint, or they may intersect.
In the former case, the agents can be handled independently. We apply Algorithm~\ref{alg:dijkstra_mod} for each agent to compute their respective budget mappings, and then determine an optimal budget allocation that minimizes the cost (egalitarian or utilitarian).

\begin{figure}[tb]
    \centering
        \includegraphics[width=.4\linewidth, page=8, trim=390 480 540 70, clip]{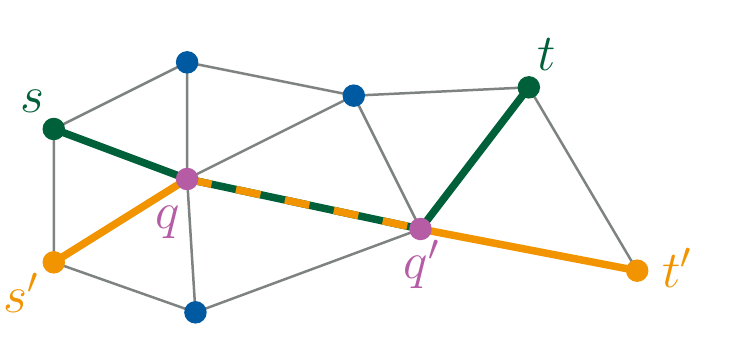}
    \caption{Pattern of $2$-\Graph{} with non-empty intersection of the two shortest paths.}
    \label{fig:inter}
\end{figure}

The latter case (the agents’ paths intersect) is more complex. 
We may assume without loss of generality that there is only a single intersection bounded by two junction vertices (these may coincide, in which case the paths only intersect at a vertex, but do not share edges).
Indeed, if two shortest paths intersect two or more times, we can just follow one of them from the first intersection point to the last intersection point.\footnote{Generalizing this observation will be crucial for dealing with 3 or more agents, see \Cref{lem:UniqueConSP,lem:HubTree}.}
Given these vertices, the overall structure breaks into (up to) five separate path branches (see \Cref{fig:inter}). 
Each branch can again be processed using Algorithm~\ref{alg:dijkstra_mod} before finding a globally optimal allocation of the budget across all branches. 

\begin{restatable}{theorem}{TwoG}
\label{theo:2-graph}
$2$-\GraphStar{} is solvable in \bigO{\beta m^2 \log m + m \beta \ell + m^2 \beta^2} time.
\end{restatable}

As in the single-agent case, the number of edges in a connected graph is 
at most $m^2$, and we can assume the total budget is at most $2m$ (a budget of $m$ for each agent). 
Hence, the runtime simplifies to $\bigO{\beta m^3}$ or $\bigO{m^4}$.

When the number of agents is viewed as a parameter, the situation is much more involved.
Our next result derives an XP algorithm for \GraphStar{}.
It is based on a structural lemma about the pairwise intersections of shortest paths.
Specifically, we will derive and make use of the fact that there exist shortest paths such that the intersection of each pair of paths is itself a (contiguous) path.
Formally, following \addauthor{Bodwin~}\citet{Bodw19a}, we say that a collection of paths $\mathcal P$ is \emph{consistent} if for any paths $P, P'\in \mathcal P$, it holds that $(V(P)\cap V(P'),E(P)\cap E(P'))$ is also a path, i.e., the paths share at most one contiguous piece.
Note that this can be the ``empty'' path if $V(P)\cap V(P') = \emptyset$, in which case the paths do not share any segment.
In other words, in a collection of consistent paths, no two paths can share a segment, then divert, and rejoin later on.
We make use of a folklore result about consistent collections of unique shortest paths \citep{Bodw19a}.

\begin{lemma}[\citet{Bodw19a}]\label{lem:UniqueConSP}
    Consider a weighted graph $G = (V,E,w)$ and a collection of $k$ terminal pairs $\{(s_i,t_i)\colon i\in [k]\}$, where $s_i,t_i\in V$. 
    If for every $i\in [k]$ the shortest $s_i$-$t_i$-paths are unique, then the collection of shortest $s_i$-$t_i$-paths is consistent.
\end{lemma}

However, we need a stronger result: even if shortest paths are not unique, there \emph{exist} consistent collections of shortest paths. 
To show this, we use well-known perturbation techniques that create unique shortest paths \citep{Char52a,DOW55a,CCE13a}, and then apply \Cref{lem:UniqueConSP}.

\begin{restatable}{lemma}{PathConsistency}\label{lem:PathConsistency}
    Given a weighted graph $G = (V,E,w)$ and a collection of $k$ terminal pairs $\{(s_i,t_i)\colon i\in [k]\}$, where $s_i,t_i\in V$, there exist a collection of paths $P_i$, $i\in [k]$, such that: 

    \begin{itemize}
        \item For every $i\in [k]$, $P_i$ is a shortest $s_i$-$t_i$-path.
        \item $\{P_i\colon i\in [k]\}$ is consistent.
    \end{itemize}
\end{restatable}

\begin{algorithm}[tb!]
    \caption{XP algorithm for \GraphStar{} parameterized by the number of agents\label{alg:XPagents}}
      \DontPrintSemicolon 
      \SetFuncSty{textsc}
      \SetKwFor{ForAll}{forall}{do}

    \KwIn{Instance $\gbpins$ of \GraphStar{} with $|\ac| = k$}
    \KwOut{Feasible solution to $\gbpins$ minimizing $c^{\star}$}

    \tcp{Budget mappings}
    \ForAll{$u\in V$}{
    Run \Cref{alg:dijkstra_mod} to obtain, for every $v\in V$ budget mapping for traveling from $u$ to $v$; denote it by $\mu_{u,v}\colon [\beta]_0\to \mathbb Q_0^+$
    }
    \tcp{Guesses}
    Guess $J\subseteq V$ with $|J| = \min\{|V|,k^2\}$\label{ln:junctions}\;
    \ForAll{$a = (s,t)\in\ac$}{
    Guess $(v^a_1,\dots, v^a_{\ell_a})\in J^{\ell_a}$ with $ \ell_a\le \min\{|V|,2k\}$ and $|\{v^a_1,\dots, v^a_{\ell_a}\}| = \ell_a$ and set $v^a_0 = s$, $v^a_{\ell_a+1} = t$, and $\tilde P_a = (\{v^a_i\colon i\in [\ell_a+1]_0\},\{\{v^a_{i-1},v^a_i\}\colon i\in [\ell_a+1]\})$\label{ln:internaljunctions}
    }
    $\mathcal E \gets \bigcup_{a\in \ac}E(\tilde P_a)$\tcp{Meta edges}
    Guess function $\tau \colon \mathcal E \to [\beta]_0$ with $\sum_{e\in \mathcal E}\tau(e) = \beta$\label{ln:BudgetToTravel}\;
    \tcp{Cost computation and selection}
    \ForAll{$a = (s,t) \in \ac$}{
    $\deltag(a) \gets \sum_{i = 1}^{\ell_a+1} \mu_{v^a_{i-1},v^a_i}(\tau(\{v^a_{i-1},v^a_i\}))$
    \label{ln:TravelTimes}}
    $c\utsup\gets \sum_{a\in \ac} \deltag(a)$, $c\egsup\gets \max_{a\in \ac} \deltag(a)$\;
    Let $\mathcal E^*$ and $\tau^*\colon \mathcal E^* \to [\beta]_0$ be meta edges and budget allocation in solution minimizing $c^{\star}$\;
    \ForAll{$\{u,v\} = e\in \mathcal E^*$}{
    Let $F_e$ be solution to $1$-\GraphStar{} when traveling from $u$ to $v$ with budget $\tau^*(e)$\label{ln:ExtractSol}
    }
    \Return $\bigcup_{e\in \mathcal E^*} F_e$ 
\end{algorithm}

We are now ready to argue that \GraphStar{} admits an XP algorithm with respect to the number of agents.
We focus on correctness here and defer the proof of the running time bound to the appendix. 

\begin{restatable}{theorem}{XPfixedK}\label{thm:XPfixedK}
    It holds that $k$-\GraphStar{} is solvable in \bigO{(2k)^{4k^2+4}m^{k^2}\beta^{3k^2}} time.
\end{restatable}

\begin{proof}[Proof of correctness]
    We apply \Cref{alg:XPagents}, whose idea is to guess the \emph{junction vertices} of agents' paths, i.e., the points where two paths split or merge.
    By \Cref{lem:PathConsistency}, some optimal solution has every pair of agents split or merge at most once, so we can assume at most $2k^2$ junction vertices, see \cref{ln:junctions}.
    Then, we guess for each agent the order of travel between junction nodes.
    We will show that an agent can pass at most $2k$ junction nodes in some optimal solution, which bounds the number of cases for this.
    Junction nodes together with edges according to the travel between them form a meta graph that describes the high-level travel structure of agents captured, where agent $a$ traverses a meta path $\tilde P_a$.
    On this meta graph, we allocate the budget to the meta edges, see \cref{ln:BudgetToTravel}.
    This corresponds to the budget allocated between two junction nodes.
    Finally, for each edge in the meta graph, we can use the budget mapping for the corresponding $1$-\GraphStar{} problem to determine the lowest cost to travel on this edge with the provided budget.
    This lets us get a travel time for each agent, see \cref{ln:TravelTimes}. 
    We obtain the total egalitarian or utilitarian welfare by taking the maximum or sum.

    We now fix the egalitarian or utilitarian objective and an instance $\ic = \gbpins$ of \Graph{} with underlying network $G = (V,E,w)$ and $|\ac| = k$.
    Let $\algc$ denote the value returned by \Cref{alg:XPagents} and let $\optc$ denote the minimum cost of a feasible solution with respect to our fixed objective.\medskip

    We will first show that $\optc \le \algc$.
    For this, consider fixed guesses for the junction nodes in $J$, junction vertices on each path $P_a$ for $a\in \ac$ giving rise to meta edges $\mathcal E$, and a budget allocation function $\tau \colon \mathcal E\to [\beta]_0$.

    For each $\{u,v\} = e\in \mathcal E$, let $F_e$ be the edges corresponding to an optimal solution to the $1$-\GraphStar{} instance when traveling from $u$ to $v$ with budget $\tau(e)$.
    Set $F = \bigcup_{e\in \mathcal E} F_e$.
    By our condition on $\tau$, we know that $F$ is feasible for $\ic$.

    Now fix an agent $(s,t) = a\in \ac$ and consider the guess for their junction vertices $v_0^a,\dots, v_{\ell_a + 1}^a$ where $v_0^a = s$ and $v_{\ell_a + 1}^a = t$.
    For every $i\in [\ell_a + 1]$, let $e = \{v_{i-1}^a,v_i^a\}$. 
    Since $F$ contains the edges in $F_e$, $a$ can travel from $v_{i-1}^a$ to $v_{i}^a$ for at most $\mu_{v^a_{i-1},v^a_i}(\tau(\{v^a_{i-1},v^a_i\}))$ whenever the edges in $F$ are upgraded.
    Hence, $\delta_F(a) \le \deltag(a)$.
    Therefore, $c^{\star}(F)$ is at most $\sum_{a\in \ac} \deltag(a)$ if $c^{\star} = c\utsup$ and at most $\max_{a\in \ac} \deltag(a)$ if $c^{\star} = c\egsup$.
    Since $\optc$ is the cost of an optimal feasible solution, we know that $\optc \le c^{\star}(F)$.
    Hence, since our combination of guesses was arbitrary, the cost of any guess can be lower bounded by $\optc$.
    It follows that $\optc \le \algc$.\medskip

    Next, we will show that $\algc\le \optc$.
    Let $F$ be a feasible solution of $\ic$ minimizing cost according to the chosen objective.
    Let $w_F\colon E \to \mathbb Q_0^+$ with $$w_F(e) = \begin{cases}
        \alpha \cdot w(e) & e\in F\\
        w(e) & e\in E\setminus F
    \end{cases}\text.$$
    This is the enhanced network with edges discounted according to $F$.
    By \Cref{lem:PathConsistency}, there exists a collection of consistent shortest paths $\mathcal P = \{P_a\colon a\in \ac\}$ on $(V,E,w_F)$.
    
    Since every pair of agents shares at most one connected sub-path, there are at most $2 \binom{k}{2} \le k^2$ junction vertices. Hence, there exists a subset $J\subseteq V$ with $|J| = \min\{|V|,k^2\}$ 
    such that whenever two paths in $\mathcal P$ join or split at a vertex $v\in V$, then $v\in J$.
    For $(s,t) = a\in \ac$ let $v_1^a,\dots, v_{\ell_a}^a$ be the at most $\ell_a \le \min\{|V|,2k\}$ vertices in $J$ that $P_a$ passes, apart from possibly $s$ and $t$.
    We assume that, traversing $P_a$ from $s$ to $t$, these vertices occur in their order of traversal and we set $v_0^a = s$ and $v_{\ell_a + 1}^a = t$.
    Note that it is possible that $s\in J$ or $t\in J$ but we only consider the internal junction points for $v_1^a,\dots, v_{\ell_a}^a$. 
    Also note that it is possible that $\ell_a = 0$, i.e., there is no such internal point.
    Define the meta path $\tilde P_a = (\{v^a_i\colon i\in [\ell_a+1]_0\},\{\{v^a_{i-1},v^a_i\}\colon i\in [\ell_a+1]\})$.

    Now consider $u,v\in J$ and let 
    \begin{align*}
        E(u,v) := &\{e\in E\colon e \text{ on sub-path from } u \text{ to } v\text{ of } P_a \\
        &\text{ for some agent } a \text{ with } \{u,v\}\in E(\tilde P_a)\}\text.
    \end{align*}
    Our next crucial insight to continue the proof is that these edge sets are disjoint.
    \begin{claim}\label{cl:budgetinducing}
        Let $u,v,u',v'\in J$ with $\{u,v\}\neq \{u',v'\}$.
        Then $E(u,v) \cap E(u',v') = \emptyset$.
    \end{claim}
    \begin{claimproof}
        Let $u,v,u',v'\in J$ with $\{u,v\}\neq \{u',v'\}$ and assume for contradiction that there exists an edge $e\in E(u,v) \cap E(u',v')$.
        Hence, there exist agents $a,a'\in \ac$ with $\{u,v\}\in E(\tilde P_a)$ and $\{u',v'\}\in E(\tilde P_{a'})$ and $e$ is on the sub-paths of $P_a$ from $u$ to $v$ and of $P_{a'}$ from $u'$ to $v'$.
        
        Recall that we have $\{u,v\}\neq \{u',v'\}$. Without loss of generality, assume that $u\notin \{u',v'\}$.
        Let $x$ be the first vertex of $e$ reached on the sub-path of $P_a$ from $u$ to $v$ and assume without loss of generality that, on $P_{a'}$, $x$ is first reached (among the endpoints of $e$) from $u'$.
        Since $u'$ and $v'$ are consecutive junction vertices on $P_{a'}$, we know that $u$ is not on the sub-path of $P_{a'}$ from $u'$ to $v'$.
        But then, since $x$ is on both paths, traversing $P_a$ from $u$ to $x$ must reach a first vertex $x'$ that is also on the sub-path from $u'$ to $v'$.
        Since $x$ is the first vertex of $e$ reached from $u$, this cannot be $v$, so this is an inner vertex on the sub-path of $P_a$ from $u$ to $v$.
        However, this means that $x'\in J$, contradicting that $u$ and $v$ are consecutive junction vertices on $P_a$.        
    \claimqed\end{claimproof}

    We can now use the insights gained in \Cref{cl:budgetinducing} to define a budget distribution.
    Define $\tau \colon \mathcal E \to [\beta]_0$ such that for $\{u,v\}=e\in\mathcal E$, we have $\tau(e)=|F\cap E(u,v)|$, i.e., we allocate the budget according to the number of enhanced edges on paths between two junction vertices.
    By \Cref{cl:budgetinducing}, we know that the sets $F\cap E(u,v)$ are pairwise disjoint, so 
    \begin{equation*}
        \sum_{e\in \mathcal E} \tau(e) = \sum_{\{u,v\} = e\in \mathcal E} |F\cap E(u,v)| \le |F| \le \beta \text.
    \end{equation*}
    
    There we use feasibility of $F$.
    Thus we have enough budget for such a budget mapping and if we have $\sum_{e\in \mathcal E} \tau(e) < \beta$, we can assign the remaining budget to arbitrary meta edges. 
    
    Hence, we now have derived suitable guesses for \Cref{alg:XPagents}.
    It remains to show that they lead to a cost estimate of at most $\optc$.
    Fix an agent $a\in \ac$ and $i\in [\ell_a+1]$.
    By definition of $E(v_{i-1}^a,v_i^a)$ and $\tau$, the sub-path of $P_a$ from $v_{i-1}^a$ to $v_i^a$ uses at most $\tau(\{v_{i-1}^a,v_i^a\})$ enhanced edges.
    Hence, since $\deltag(a)$ sums the minimum costs for traveling between junction vertices with allocated budget according to $\tau$, we have that $\deltag(a) \le w_F(P_a)$.
    As a result, these guesses lead to a (utilitarian or egalitarian) cost of at most $\optc$.
    Since \Cref{alg:XPagents} returns the minimum cost among all guesses, we conclude that $\algc \le \optc$.
    This completes the proof of correctness of \Cref{alg:XPagents}.
\proofqed\end{proof}

A natural follow-up question is whether \GraphStar{} admits an FPT algorithm.
We will now argue that this is unlikely, by showing that \GraphDec{} is W[1]-hard parameterized by the number of agents.
The proof proceeds by a reduction from \MCClique.

\begin{restatable}{theorem}{WoneGraph}\label{thm:W1Graph}
    Let $\alpha\in [0,1)$.
    Then, $\alpha$-\GraphDec{} is W[1]-hard parameterized by the number of agents, even if the budget is an additional parameter.
\end{restatable}

Interestingly, the hardness result of \Cref{thm:W1Graph}  holds even if the number of agents and the budget are simultaneously a parameter.
In contrast, our problem is in XP with respect to each of these parameters. For the number of agents, this is shown by 
\Cref{alg:XPagents}. For the budget, 
we can simply guess $\beta$ edges whose costs are to be reduced by a factor of $\alpha$, and then compare the cost of all guesses.
This yields an XP algorithm as there are $\binom{\ell}{\beta} \le \ell^\beta$ guesses.

\Cref{thm:W1Graph} also raises the question of whether
there is a natural structured subdomain of \Graph{}
that admits an FPT algorithm.
We will now show that this is indeed the case for \GraphHub{}, i.e., the class of instances where all agents travel to the same hub terminal.
Our proof relies on a further insight about shortest path systems when all paths terminate in the same vertex: 
If such a system is consistent, then it must form a tree.

\begin{restatable}{lemma}{HubTree}\label{lem:HubTree}
    Consider a weighted graph $G = (V,E,w)$ and let $k\in \mathbb N$, $s_i\in V$ for all $i\in [k]$ and $t\in V$. 
    For each $i\in [k]$, let $P_i$ be a shortest $s_i$-$t$-path.
    Define $V_T = \bigcup_{i\in [k]}V(P_i)$ and  $E_T = \bigcup_{i\in [k]}E(P_i)$.
    If $\{P_i\colon i\in [k]\}$ is consistent, then
    $(V_T,E_T)$ is a tree.
\end{restatable}

\Cref{lem:HubTree} enables us to derive our fixed-parameter tractability result for \GraphHubStar{}.
Specifically, we develop a dynamic program that uses \Cref{alg:dijkstra_mod} as a preliminary step.
The dynamic program then obtains a value for every possible subset of agents, terminal hub, and budget that captures the minimum travel cost for the given subset of agents to reach the designated hub under the given budget constraint.
Our approach is somewhat reminiscent of the Dreyfus--Wagner algorithm for computing a minimum weight Steiner tree when the number of terminals (i.e., vertices that have to be connected by the Steiner tree) is the parameter \citep{DrWa71a,CFK+15b}.
However, it needs several sophisticated enhancements to account for the budget, and parts of the proof require a separate treatment of the utilitarian and egalitarian welfare.
To obtain our bound on the running time, we enhance our basic dynamic program by computing an intermediate table that accounts for finding an optimal way of merging two subsets of agents at a specific node.
In particular, this saves an unnecessary blow-up in running time due to a threefold splitting of the budget.

\begin{restatable}{theorem}{FPThub}\label{thm:FPThub}
    It holds that $k$-\GraphHubStar{} is solvable in time $\mathcal O(m^3 \beta + 3^k m\beta^2 + 2^k m^2\beta^2)$.
\end{restatable}

The key difficulty with extending our method beyond instances in \GraphHub{} is that consistent shortest path systems can be cyclic.
While this does not happen for $2$ agents, it may occur for $3$ agents traveling between the vertices of a triangle.
Hence, without the common hub assumption, in particular \Cref{lem:HubTree} breaks down for at least $3$ agents.

\begin{example}
    Consider a graph $G = (V,E,w)$ with $V = \{v_1,v_2,v_3\}$, $E = \{\{v_1,v_2\},\{v_1,v_3\},\{v_2,v_3\}\}$, and $w(e) = 1$ for all $e\in E$. Let $k=3$.
    For $i\in [k]$, let $s_i = v_i$ and $t_i = v_{i+1}$ (in this example, we interpret indices modulo~$3$, so that $v_4:=v_1$).
    Then, for $i\in [k]$, let $P_i$ the shortest $s_i$-$t_i$-path, i.e., the path that simply uses the edge $\{s_i,t_i\}$.
    Clearly, $\{P_1,P_2,P_3\}$ is consistent ($P_i$ and $P_{i+1}$ share exactly $v_{i+1}$).
    However, the union of these three paths is $G$, which is not a tree.
\end{example}

\subsection{Variable Number of Agents}

He et al.~\citet{HBL+24} 
prove that both egalitarian and utilitarian cost minimization in their railway‐network model are $\np$‐complete, via a reduction from \Sat{}.
The same reduction applies to \Graph{}, but only when $\alpha=0$. 
To extend this hardness result to any fixed $\alpha\in [0, 1)$, we provide a simple reduction from \SetCover{}.

\begin{restatable}{theorem}{GBPhardness}
\label{cor:graph0}
Let $\alpha \in [0,1)$. Then $\alpha$-\GraphDec{} is \np-complete.
\end{restatable}

A natural question is whether \np‐hardness also rules out nontrivial approximation. When $\alpha=0$, for a suitably selected high budget one can achieve zero cost for all agents. 
We show that designing a nontrivial approximation algorithm for this setting is as difficult as designing an exact algorithm, and therefore
no polynomial‐time approximation exists. For the complete proof, we refer to \Cref{pro:GBPInap}.

\begin{restatable}{theorem}{GBPInap}
\label{cor:appa}
Let $\alpha = 0$. Then
$\alpha$-\GraphStar{} is inapproximable.
\end{restatable}

Moreover, the egalitarian version resists close approximation even for positive~$\alpha$. 
By analyzing our \SetCover‐based reduction more closely, see \Cref{pro:GBPApproxBound}, 
we prove that no polynomial‐time algorithm can guarantee an approximation ratio below $(\alpha+1)/(2\alpha)$. 
We note that this result is only obtained for the egalitarian welfare; inapproximability for the utilitarian welfare and positive $\alpha$ remains open.

\begin{restatable}{theorem}{GBPApproxBound}
\label{theo:optim:3}
Let $\alpha \in (0,1)$ and $\gamma < {} \frac{\alpha + 1}{2\alpha}$.
Then $\alpha$-\GraphMax{} is $\gamma$-inapproximable. 
\end{restatable}

Interestingly, our lower‐bound construction makes use of instances where every agent’s cost is exactly $2\alpha$ in the “Yes’’ case, but at least $(\alpha+1)$ in the “No’’ case---hence the ratio $(\alpha+1)/(2\alpha)$. 
Note that $\frac{\alpha + 1}{2\alpha}$ is unbounded for $\alpha$ tending to~$0$ and converges to~$1$ when $\alpha$ converges to~$1$.
Hence, while the inapproximability is not severe for large enough $\alpha$, it can be arbitrarily bad.
Notably, for the utilitarian objective, the same construction does not translate into hardness: when the number of agents is large, a single agent’s elevated cost only has a vanishing effect on the total, not giving the same gap.

\section{Conclusion}

We have investigated a stylized model of investment in public transportation infrastructure, in which a central authority needs to allocate a limited budget to reduce the costs of some of the network edges.
We are interested in two objectives: minimizing the total travel time, as captured by the utilitarian welfare, and minimizing the maximum travel time among all agents, as captured by the egalitarian welfare.

\renewcommand{\arraystretch}{1.2}

\begin{table}[t!]
\caption{Summary of results.
Here, $\alpha$ is the discount factor, $\beta$ is the budget, and $m$ and $\ell$ are the number of vertices and edges of the underlying graph. 
\label{tab:summary}}
\centering
\resizebox{1\linewidth}{!}{%
\begin{tabular}{|l|c|c|c|}
\hline
\textbf{Problem} & \textbf{Results} & \phantom{     }\textbf{Reference}\phantom{     } & \textbf{Restrictions} \\
\Xhline{1.2pt}
1-\GraphStar & solvable in \bigO{\beta m\log m + \beta \ell} & Theorem~\ref{theo:1-graph} & \\
\hline
2-\GraphStar & solvable in \bigO{\beta m^2 \log m + m \beta \ell + m^2 \beta^2} & Theorem~\ref{theo:2-graph} & \\
\hline
$k$-\GraphStar & solvable in \bigO{(2k)^{4k^2+4}m^{k^2}\beta^{3k^2}} & Theorem~\ref{thm:XPfixedK} & \\
\hline
$k$-\GraphHubStar & solvable in $\mathcal O(m^3 \beta + 3^k m\beta^2 + 2^k m^2\beta^2)$ & Theorem~\ref{thm:FPThub} & \\
\hline
\multirow{2}{*}{$\alpha$-\GraphDec}
& \np-complete & Theorem~\ref{cor:graph0} & \\
\cline{2-4}
& W[1]-hard param. by number of agents and budget & Theorem~\ref{thm:W1Graph} & \\
\hline
\multirow{2}{*}{$\alpha$-\GraphStar}
& inapproximable & Theorem~\ref{cor:appa} & $\alpha = 0$\\
\cline{2-4}
& $\alpha^{-1}$-approximable & Observation~\ref{obs:TrivApprox} &  $\alpha\neq 0$\\ 
\hline
$\alpha$-\GraphMax & not $\gamma$-approximable & Theorem~\ref{theo:optim:3} & $\alpha\neq 0, \gamma < \frac{\alpha + 1}{2\alpha}$ \\
\hline
\end{tabular}%
} 
\end{table}

We obtained algorithms and intractability results for the associated optimization and decision problems.
Our contributions are summarized in Table~\ref{tab:summary}. 
Our first algorithm, \Cref{alg:dijkstra_mod}, is a polynomial-time algorithm for the \gbp{} with a single agent.
It is obtained by extending Dijkstra's algorithm to pivot pairs, which encode a vertex and a budget level.
This algorithm becomes an important subroutine in all subsequent algorithms.
For two agents, we apply it for each of the segments of the two travel topologies that can occur (with or without an intersection).
We extend this idea to an arbitrary number of agents, thereby obtaining an XP algorithm. 
This algorithm first guesses a general travel topology and then applies \Cref{alg:dijkstra_mod} on all segments.
Notably, the guesses  implicitly impose a structural assumption on the number of intersections in the travel paths of pairs of agents.
This assumption can be made without loss of generality, because, as we show, it is satisfied by some optimal solution.

Next, we obtain a W[1]-hardness result, which makes it unlikely that the XP algorithm can be improved to an FPT algorithm.
Nonetheless, we obtain an FPT algorithm under the domain restriction where all agents travel to a common hub.
Finally, we conclude with hardness and inapproximability results when the number of agents is variable.

Interestingly, with the single exception of \Cref{theo:optim:3}, all results hold for both utilitarian and egalitarian welfare.
For a single agent, this is trivial as both notions of welfare coincide.
For our XP algorithm, this is the case because we can compute the welfare for each guess and then compare them.
For our FPT algorithm, we have to make a distinction for the update formulas of our dynamic programs based on the welfare.
For our hardness results, whenever they work for both welfare notions, they utilize the same construction (albeit with different target costs), which demonstrates the robustness of the constructions.
Moreover, our W[1]-hardness and \np-completeness results hold whenever $\alpha$ assumes any fixed value in $[0,1)$.

Improving public transport will be an essential component of transforming the modern society,  which needs to deal with  unprecedented challenges such as climate change.
We believe that studying simple models serves as an important theoretical foundation for a more complex discourse.

There are multiple research directions to build upon our work.
First, there is a gap between our inapproximability results and the trivial $\alpha^{-1}$-approximation.
It would be interesting to see if one can achieve a better approximation guarantee in polynomial time.
For our negative results, it would be intriguing to investigate the possibility of circumventing them in restricted domains based on networks that are likely to occur in practice, e.g., planar graphs.
Finally, it would be interesting to enhance our simple model of public transport by incorporating further features, such as directed travel, agents arriving online, or varying per-edge improvement costs, or by considering additional solution concepts such as proportional fairness and core stability.

\clearpage
\appendix
\section*{Appendix}

In this appendix, we provide omitted material as well as a comparison with the \rwp{} by \addauthor{He et al.~}\citet{HBL+24}.

\section{Performance of Greedy Algorithms}
\label{count:greedy}

In this appendix, we consider the performance of $\gdyup$ and $\gdydn$.

\gdpperformance*

\begin{proof}
    Let $\alpha \in [0,1)$ and consider a fixed budget $\beta \in \mathbb{N}$.
    We define an instance 
    $\gbpins$ of $\Graph{}$ along the illustration in \Cref{fig:greedy:1}.

    The weighted graph $G = (V,E,w)$ is given by 
    $$V = \bigcup_{i\in [\beta]_0} \left(\{s_i,t_i\}\cup \{v_i^j\colon j\in [\beta]\}\right)\cup \{q_i\colon i\in [\beta]\}\cup \{q,q'\}$$
    and 
    \begin{align*}
        E &= \{\{s_i,q\},\{t_i,q'\},\{s_i,v_i^1\},\{v_i^\beta,t_i\}\colon i\in [\beta]_0\}\\
        &\cup \{\{v_i^j,v_i^{j+1}\}\colon  i\in [\beta]_0, j\in [\beta-1]\}\\
        &\cup \{\{q,q_1\},\{q_\beta,q'\}\}\cup \{\{q_i,q_{i+1}\}\colon i\in [\beta-1]\}\text.
    \end{align*}
    
    Consider any strictly increasing sequence $(\epsilon_i)_{i\ge 0}$ with $\epsilon_0 > 0$ and, for all $i \ge 0$, 
    \begin{equation*}
        \epsilon_i < 
        \begin{cases}
            \min\left\{1,\frac 1{\alpha} - 1\right\} & \text{if } \alpha > 0\\
            1 & \text{if } \alpha = 0
        \end{cases}\text.
    \end{equation*}
    Note that this is well-defined for $\alpha \in (0,1)$.
    For all $i\in [\beta]_0$, we set $w(s_i,v_i^1) = 1 + \epsilon_i$.
    Moreover, we set $w(e) = 1$ for all other edges $e$.
    
    We have $\ac = \{a_i\colon i\in [\beta]_0\}$ with $a_i = (s+i,t+i)$, i.e., there are $\beta + 1$ agents.

    The idea is that agents shortest paths are the direct routes between $s_i$ and $t_i$ along vertices $v_i^j$, whereas all agents would benefit from using the ``motorway'' between $q$ and $q'$.
    Both variants of the greedy algorithms will eventually select edges along direct routes, which are not shared between agents, because these offer locally optimal gains.
    As a result, the final cost for each agent is similar to the cost of not investing in edges at all. 
    In contrast, the optimal solution is to reduce $\beta$ of the $\beta + 1$ edges on the motorway, essentially leaving only the access and exit costs. 
    As $\beta$ grows, these additional costs become negligible.

    Let us now formally consider the behavior of $\gdyup$.
    Each agent $a_i = (s_i, t_i)$ has a direct route consisting of a single edge of weight $1+\epsilon_i$ followed by $\beta$ edges of weight $1$. 
    Hence, the total cost for using the direct route is $1 + \epsilon_i + \beta \in [\beta +1,\beta +2]$.
Note that this is the shortest route for each agent.
Each other connection between $s_i$ and $t_i$ has at least length $\beta + 3$, so their distance would not decrease even if any edge on any other path was removed.
We now show by induction that $\gdyup$ adds the edges $\{s_\beta,v_\beta^1\}$, dots, $\{s_1,v_1^1\}$ in this order.
Assume that the $i\ge 0$ edges $S_i = \{\{s_\beta,v_\beta^1\},\dots, \{s_{\beta-i+1},v_{\beta-i+1}^1\}\}$, which is the empty set for $i = 0$. 
Then, the current cost for agents $a_0, \dots, a_{\beta - i}$ is still to use their direct connection.
We now reason about what edge $\gdyup$ should select next.
Selecting an edge on the motorway does not decrease the cost for any agent.
Moreover, selecting an edge on the direct route of an agent does not reduce the cost of any other agent.
Indeed, then they would have to travel to the start and from the end of the other agents direct route via $q$ and $q'$, incurring a cost of at least $\beta + 4$ (bounded by the case where the heavy edge on this agent's route that has already been reduced to a weight of~$0$).
Similarly, agents $a_{\beta - i+1},\dots, a_{\beta}$ cannot gain from selecting any edge not on their direct route.
Hence, $\gdyup$ will select an edge on the direct route of some agent and only one agent will benefit from that.
For the utilitarian cost, the next selected edge is the edge of heaviest weight, which is $\{\{s_{\beta-i},v_{\beta-i}^1\}$.
For the egalitarian cost, note that the current costs are 

\begin{equation*}
    \delta_{S_i}(a_j)
    \begin{cases}
        \beta + 1 + \epsilon_j & j = 1,\dots, \beta - i\\
        \beta + \alpha(1+\epsilon_j) & j = \beta - i + 1, \dots, \beta
    \end{cases}\text.
\end{equation*}

Note that for $j = \beta - i + 1, \dots, \beta$, it holds that $\alpha (1 + \epsilon_j) < \alpha (1 + (\frac 1{\alpha}-1)) = 1$.
Hence, the agent with the highest cost is $a_{\beta - i}$ and we reduce the egalitarian cost the most by selecting an edge on their direct route.
Without loss of generality, $\gdyup$ selects $\{s_{\beta+i},v_{\beta+i}^1\}$ next (which minimizes the utilitarian cost locally subject to minimizing the egalitarian cost).
In conclusion, $\gdyup$ will return $S_{\uparrow} = \{\{s_i,v_i^1\}\colon i\in [\beta]\}$ under both cost functions.

\begin{figure}[tb]
\centering
\includegraphics[width=.6\textwidth, page=4, trim=30 150 550 70, clip]{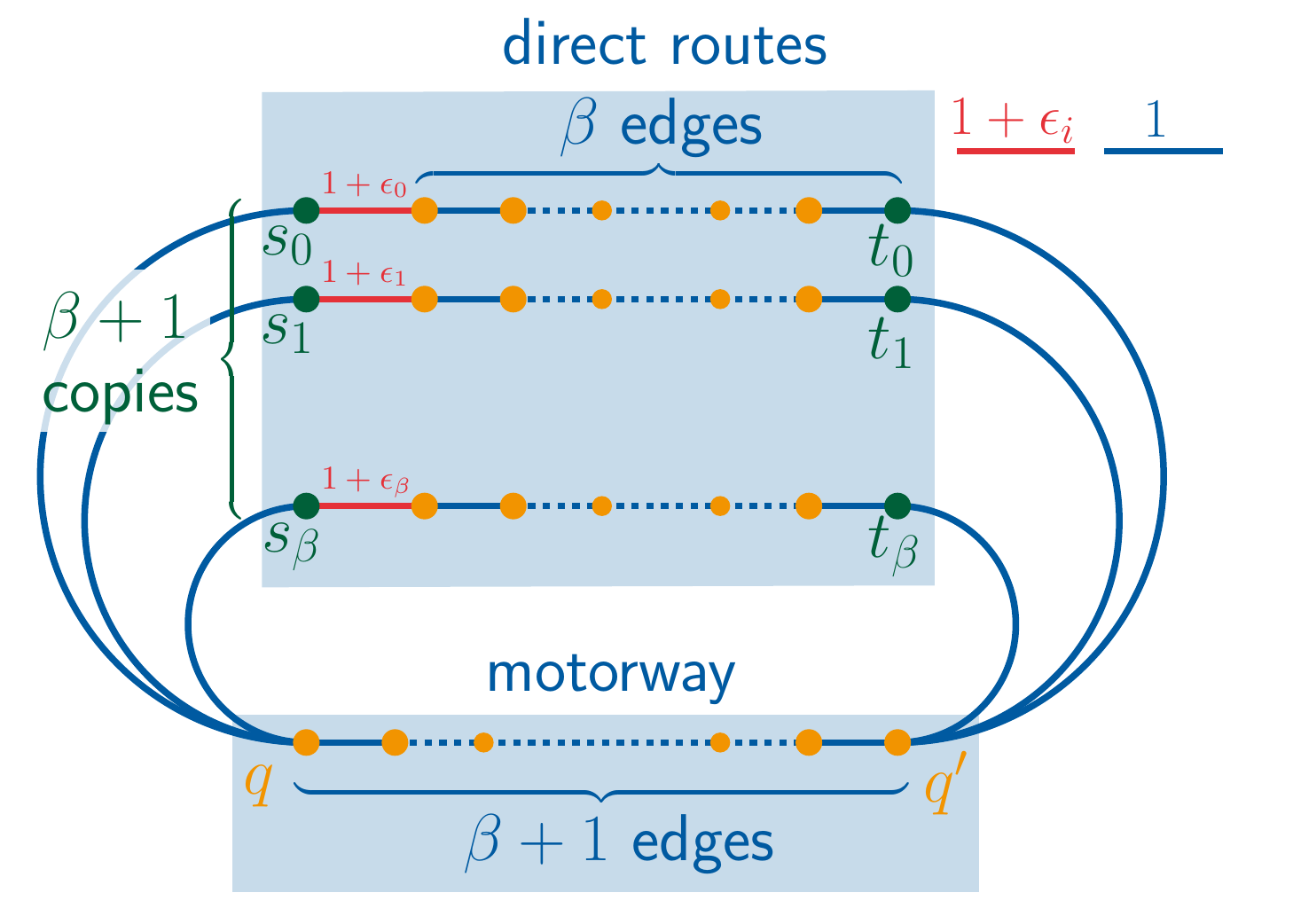}
 \caption{Illustration of instances constructed in \Cref{prop:gdy}. 
 The greedy algorithms yield an arbitrarily suboptimal solution.
 \label{fig:greedy:1}}
\end{figure}

Next, we consider the performance  of $\gdyup$ on this instance.
All shortest paths when no edges are discounted are also shortest paths when all edges are discounted (and in the case $\alpha = 0$, all paths are shortest paths).
Hence, $\gdydn$ can start by deleting edges in a way such that at least one shortest path survives for each agent.
Therefore, $\gdydn$ my start by deleting all edges apart from the direct routes.
At this point, the next deleted edge will affect the cost of the solution.
Note that, as long as the difference in the number of edge included in different agents shortest paths is at most one, each agent's shortest path remains their unique shortest path.
Hence, for the utilitarian cost, we can next delete the edges in the sets $\{\{v_i^j,v_i^{j+1}\}\colon i\in [\beta]_0\}$ for $j\in [\beta - 1]$, and $\{\{v_i^\beta,t_i\}i\in [\beta]_0\}$.
We process these sets in an arbitrary order and, in turn, the edges within each set in an arbitrary order.
The same is possible for the egalitarian cost.
The final edge to be deleted is $\{s_0,v_0^1\}$ because this is the lightest remaining edge and, therefore, leads to the smallest increase in egalitarian or utilitarian welfare.
Hence, $\gdydn$ returns $S_{\downarrow} = S_{\uparrow}$.

We now evaluate the performance of both algorithms by comparing the cost of their outcome with a solution that selects $\beta$ edges on the motorway.
Since $\gdyup$ returns the same solution as $\gdydn$, we can restrict attention to $\gdyup$.

First, note that 
    \begin{align}
        c\egsup(S_{\uparrow}) &= \beta + (1+\epsilon_0)> \beta \text{ and}\notag\\
        c\utsup(S_{\uparrow}) &= \beta + (1+\epsilon_0) + \sum_{i = 1}^{\beta} \beta + \alpha(1+\epsilon_i) > (1+\beta)\beta\text.\label{eq:utgdy}
    \end{align}

Now, consider the set $S^* = \{\{q_i,q_{i+1}\}\colon i\in [\beta - 1]\}\cup \{\{q,q_1\}\}$.
It holds that 
\begin{equation*}
    c\egsup(S^*) \le 3 + \alpha \beta \quad\text{and}\quad c\utsup(S^*) \le (\beta+1)(3 + \alpha \beta)\text,
\end{equation*}
where we bound with the case where all agents use the motorway.

Combining this with \Cref{eq:utgdy}, it follows that 
\begin{equation*}
    \frac {c\egsup(S^*)}{c\egsup(S_{\uparrow})} < \frac{3 + \alpha \beta}{\beta} \rightarrow \alpha \text{ for }\beta \to \infty\text.
\end{equation*}
and 
\begin{equation*}
    \frac {c\utsup(S^*)}{c\utsup(S_{\uparrow})} < \frac{(\beta+1)(3 + \alpha \beta)}{(1+\beta)\beta} \rightarrow \alpha \text{ for }\beta \to \infty\text.
\end{equation*}

Hence, $\gdyup$ is no $\gamma$-approximation algorithm for $\GraphStar{}$ for any $\gamma < \frac 1{\alpha}$.
\proofqed\end{proof}

\section{Single-Agent Case}\label{app:oneag}

In this appendix, we provide additional material concerning $1$-\GraphStar{}.

\subsection{Example Run of Algorithm~\ref{alg:dijkstra_mod}}

\begin{figure}[tb!]
\vspace*{-0.2cm}
    \centering
\includegraphics[width= 0.6\textwidth, page=7, trim=30 490 715 85, clip]{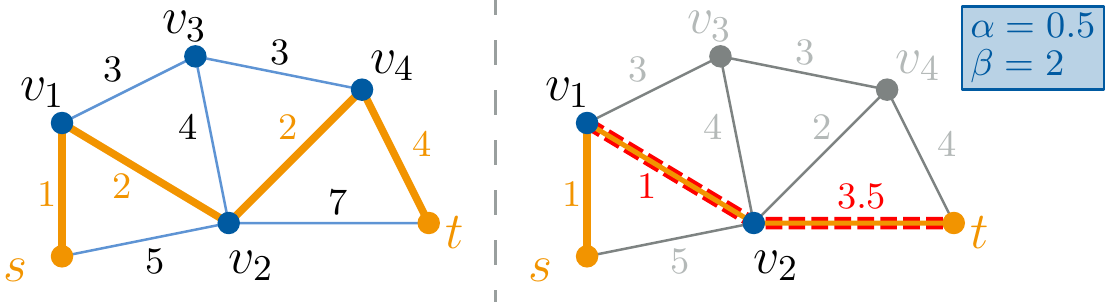} 
    \caption{Instance for illustrating the execution of the algorithm.
	We display the shortest paths without (left) and with optimal (right) investment.
    \label{fig:graphsum:c}}
\end{figure}

\begin{example}
We demonstrate \Cref{alg:dijkstra_mod} for the instance in \Cref{fig:graphsum:c} with $\alpha = 0.5$ and $\beta = 2$. 
We summarize the iterations of the algorithm in \Cref{tab:graphsum:1}.

We begin in Iteration~$I$ where pivot pair $(s,0)$ has distance~$0$. 
Both neighbors $v_{1}$ and $v_{2}$ receive two updates each, one non‐reducing update $D[v_{1},0]=1$, $D[v_{2},0]=5$ and one reducing update $D[v_{1},1]=0.5$, $D[v_{2},1]=2.5$. 
There is no budget‐increasing update in the first iteration.

In Iteration~$II$, the smallest distance of members in $Q$ is attained by $(v_1,1)$ with $D(v_1,1) = 0.5$, so the pivot pair is $(v_{1},1)$. 
For $v_2$, we have $D[v_1,1] + w(v_1,v_2) = 2.5$, so we perform no non-increasing update. 
However, we update neighbors' distances as $D[v_{2},2]=1.5$, $D[v_{3},1]=3.5$, and $D[v_{3},2]=2$. 
In addition, the budget‐increasing update sets $D[v_{1},2]=0.5$.

In Iteration~$III$, pivot $(v_{1},2)$ produces no further updates. 
By Iteration~$VI$, both $(v_{2},1)$ and $(v_{3},2)$ have distance $2$, and the algorithm prioritizes $(v_{2},1)$ due to its smaller budget indicator.

We continue in this manner until Iteration~$XIV$, when the target vertex $t$ is pivoted for the first time and $D[t,2]=5.5$ is recorded, which is the minimum distance from $s$ to $t$ with budget at most $2$. 
	Note that if we proceed the algorithm, we would further update $D(t,0)$, but this cannot lead to a further decrease of $D[t,2]$.
\hfill$\lhd$
\end{example}

\begin{table}[ht!]
\caption{Exemplary execution of Algorithm~\ref{alg:dijkstra_mod}.  
  Columns and rows represent budget indicators and vertices, respectively.  
  The initial routing pair $(s,0)$ is omitted. 
  Cells contain the current values $D(v,b)$. 
  Blank cells mean a value of $\infty$.  
  Changes are marked in blue, and the pivot pair is underlined.  
  Iterations X to XIV are merged since no further updates occur. 
  \vspace{1em}
  \label{tab:graphsum:1}}
  \centering
  \noindent\makebox[0.4\textwidth]{%
  \begin{tabular}{|p{32pt}<{\centering}|*{3}{p{18pt}<{\centering}|}}
      \hline
        \rowcolor[gray]{.9}
        I & 
        $\hspace{-3pt}b\hspace{-2pt}=\hspace{-2pt}0\hspace{-3pt}$ & 
        $\hspace{-3pt}b\hspace{-2pt}=\hspace{-2pt}1\hspace{-3pt}$ & 
        $\hspace{-3pt}b\hspace{-2pt}=\hspace{-2pt}2\hspace{-3pt}$ \\ \hline
        \cellcolor[gray]{.9} $v_1$ & \textcolor{blue}{1} & \textcolor{blue}{0.5} &  \\ \hline
        \cellcolor[gray]{.9} $v_2$ & \textcolor{blue}{5} & \textcolor{blue}{2.5} &  \\ \hline
        \cellcolor[gray]{.9} $v_3$ &  &  &  \\ \hline
        \cellcolor[gray]{.9} $v_4$ &  &  &  \\ \hline
        \cellcolor[gray]{.9} $t$ &  &  &  \\ \hline
    \end{tabular}
  \begin{tabular}{|p{32pt}<{\centering}|*{3}{p{18pt}<{\centering}|}}
      \hline
        \rowcolor[gray]{.9}
        II & 
        $\hspace{-3pt}b\hspace{-2pt}=\hspace{-2pt}0\hspace{-3pt}$ & 
        $\hspace{-3pt}b\hspace{-2pt}=\hspace{-2pt}1\hspace{-3pt}$ & 
        $\hspace{-3pt}b\hspace{-2pt}=\hspace{-2pt}2\hspace{-3pt}$ \\ \hline
        \cellcolor[gray]{.9} $v_1$ & 1& \underline{0.5} & \textcolor{blue}{0.5} \\ \hline
        \cellcolor[gray]{.9} $v_2$ & 5 & 2.5 & \textcolor{blue}{1.5} \\ \hline
        \cellcolor[gray]{.9} $v_3$ &  & \textcolor{blue}{3.5} & \textcolor{blue}{2} \\ \hline
        \cellcolor[gray]{.9} $v_4$ &  &  &  \\ \hline
        \cellcolor[gray]{.9} $t$ &  &  &  \\ \hline
    \end{tabular}
            }%
        \vspace*{2pt}

  \noindent\makebox[0.4\textwidth]{%
  \begin{tabular}{|p{32pt}<{\centering}|*{3}{p{18pt}<{\centering}|}}
      \hline
        \rowcolor[gray]{.9}
        III & 
        $\hspace{-3pt}b\hspace{-2pt}=\hspace{-2pt}0\hspace{-3pt}$ & 
        $\hspace{-3pt}b\hspace{-2pt}=\hspace{-2pt}1\hspace{-3pt}$ & 
        $\hspace{-3pt}b\hspace{-2pt}=\hspace{-2pt}2\hspace{-3pt}$ \\ \hline
        \cellcolor[gray]{.9} $v_1$ & 1& 0.5 & \underline{0.5} \\ \hline
        \cellcolor[gray]{.9} $v_2$ & 5 & 2.5 & 1.5 \\ \hline
        \cellcolor[gray]{.9} $v_3$ & & 3.5 & 2 \\ \hline
        \cellcolor[gray]{.9} $v_4$ &  &  &  \\ \hline
        \cellcolor[gray]{.9} $t$ &  &  &  \\ \hline
    \end{tabular}

  \begin{tabular}{|p{32pt}<{\centering}|*{3}{p{18pt}<{\centering}|}}
      \hline
        \rowcolor[gray]{.9}
        IV & 
        $\hspace{-3pt}b\hspace{-2pt}=\hspace{-2pt}0\hspace{-3pt}$ & 
        $\hspace{-3pt}b\hspace{-2pt}=\hspace{-2pt}1\hspace{-3pt}$ & 
        $\hspace{-3pt}b\hspace{-2pt}=\hspace{-2pt}2\hspace{-3pt}$ \\ \hline
        \cellcolor[gray]{.9} $v_1$ & \underline{1}& 0.5 & 0.5 \\ \hline
        \cellcolor[gray]{.9} $v_2$ & \textcolor{blue}{3} & \textcolor{blue}{2} & 1.5 \\ \hline
        \cellcolor[gray]{.9} $v_3$ & \textcolor{blue}{4} & \textcolor{blue}{2.5} & 2 \\ \hline
        \cellcolor[gray]{.9} $v_4$ &  &  &  \\ \hline
        \cellcolor[gray]{.9} $t$ &  &  &  \\ \hline
    \end{tabular}
            }%
        \vspace*{2pt}

  \noindent\makebox[0.4\textwidth]{%
  \begin{tabular}{|p{32pt}<{\centering}|*{3}{p{18pt}<{\centering}|}}
      \hline
        \rowcolor[gray]{.9}
        V & 
        $\hspace{-3pt}b\hspace{-2pt}=\hspace{-2pt}0\hspace{-3pt}$ & 
        $\hspace{-3pt}b\hspace{-2pt}=\hspace{-2pt}1\hspace{-3pt}$ & 
        $\hspace{-3pt}b\hspace{-2pt}=\hspace{-2pt}2\hspace{-3pt}$ \\ \hline
        \cellcolor[gray]{.9} $v_1$ & 1& 0.5 & 0.5 \\ \hline
        \cellcolor[gray]{.9} $v_2$ & 3 & 2 & \underline{1.5} \\ \hline
        \cellcolor[gray]{.9} $v_3$ & 4 & 2.5 & 2 \\ \hline
        \cellcolor[gray]{.9} $v_4$ &  &  & \textcolor{blue}{3.5} \\ \hline
        \cellcolor[gray]{.9} $t$ &  &  & \textcolor{blue}{8.5} \\ \hline
    \end{tabular}
  \begin{tabular}{|p{32pt}<{\centering}|*{3}{p{18pt}<{\centering}|}}
      \hline
        \rowcolor[gray]{.9}
        VI & 
        $\hspace{-3pt}b\hspace{-2pt}=\hspace{-2pt}0\hspace{-3pt}$ & 
        $\hspace{-3pt}b\hspace{-2pt}=\hspace{-2pt}1\hspace{-3pt}$ & 
        $\hspace{-3pt}b\hspace{-2pt}=\hspace{-2pt}2\hspace{-3pt}$ \\ \hline
        \cellcolor[gray]{.9} $v_1$ & 1& 0.5 & 0.5 \\ \hline
        \cellcolor[gray]{.9} $v_2$ & 3 & \underline{2} & 1.5 \\ \hline
        \cellcolor[gray]{.9} $v_3$ & 4 & 2.5 & 2 \\ \hline
        \cellcolor[gray]{.9} $v_4$ &  & \textcolor{blue}{4} & \textcolor{blue}{3} \\ \hline
        \cellcolor[gray]{.9} $t$ &  & \textcolor{blue}{9} & \textcolor{blue}{5.5} \\ \hline
    \end{tabular}
        }%
        \vspace*{2pt}

  \noindent\makebox[0.4\textwidth]{%
  \begin{tabular}{|p{32pt}<{\centering}|*{3}{p{18pt}<{\centering}|}}
      \hline
        \rowcolor[gray]{.9}
        VII & 
        $\hspace{-3pt}b\hspace{-2pt}=\hspace{-2pt}0\hspace{-3pt}$ & 
        $\hspace{-3pt}b\hspace{-2pt}=\hspace{-2pt}1\hspace{-3pt}$ & 
        $\hspace{-3pt}b\hspace{-2pt}=\hspace{-2pt}2\hspace{-3pt}$ \\ \hline
        \cellcolor[gray]{.9} $v_1$ & 1& 0.5 & 0.5 \\ \hline
        \cellcolor[gray]{.9} $v_2$ & 3 & 2 & 1.5 \\ \hline
        \cellcolor[gray]{.9} $v_3$ & 4 & 2.5 & \underline{2} \\ \hline
        \cellcolor[gray]{.9} $v_4$ &  & 4 & 3 \\ \hline
        \cellcolor[gray]{.9} $t$ &  & 9 & 5.5 \\ \hline
    \end{tabular}
  \begin{tabular}{|p{32pt}<{\centering}|*{3}{p{18pt}<{\centering}|}}
      \hline
        \rowcolor[gray]{.9}
        VIII & 
        $\hspace{-3pt}b\hspace{-2pt}=\hspace{-2pt}0\hspace{-3pt}$ & 
        $\hspace{-3pt}b\hspace{-2pt}=\hspace{-2pt}1\hspace{-3pt}$ & 
        $\hspace{-3pt}b\hspace{-2pt}=\hspace{-2pt}2\hspace{-3pt}$ \\ \hline
        \cellcolor[gray]{.9} $v_1$ & 1& 0.5 & 0.5 \\ \hline
        \cellcolor[gray]{.9} $v_2$ & 3 & 2 & 1.5 \\ \hline
        \cellcolor[gray]{.9} $v_3$ & 4 & \underline{2.5} & 2 \\ \hline
        \cellcolor[gray]{.9} $v_4$ &  & 4 & 3 \\ \hline
        \cellcolor[gray]{.9} $t$ &  & 9 & 5.5 \\ \hline
    \end{tabular} 
            }%
        \vspace*{2pt}

  \noindent\makebox[0.4\textwidth]{%
  \begin{tabular}{|p{32pt}<{\centering}|*{3}{p{18pt}<{\centering}|}}
      \hline
        \rowcolor[gray]{.9}
        IX & 
        $\hspace{-3pt}b\hspace{-2pt}=\hspace{-2pt}0\hspace{-3pt}$ & 
        $\hspace{-3pt}b\hspace{-2pt}=\hspace{-2pt}1\hspace{-3pt}$ & 
        $\hspace{-3pt}b\hspace{-2pt}=\hspace{-2pt}2\hspace{-3pt}$ \\ \hline
        \cellcolor[gray]{.9} $v_1$ & 1& 0.5 & 0.5 \\ \hline
        \cellcolor[gray]{.9} $v_2$ & \underline{3} & 2 & 1.5 \\ \hline
        \cellcolor[gray]{.9} $v_3$ & 4 & 2.5 & 2 \\ \hline
        \cellcolor[gray]{.9} $v_4$ & \textcolor{blue}{5} & 4 & 3 \\ \hline
        \cellcolor[gray]{.9} $t$ & \textcolor{blue}{10} & \textcolor{blue}{6.5} & 5.5 \\ \hline
    \end{tabular}
  \begin{tabular}{|p{32pt}<{\centering}|*{3}{p{18pt}<{\centering}|}}

      \hline
        \rowcolor[gray]{.9}
        \small{X-XIV} & 
        $\hspace{-3pt}b\hspace{-2pt}=\hspace{-2pt}0\hspace{-3pt}$ & 
        $\hspace{-3pt}b\hspace{-2pt}=\hspace{-2pt}1\hspace{-3pt}$ & 
        $\hspace{-3pt}b\hspace{-2pt}=\hspace{-2pt}2\hspace{-3pt}$ \\ \hline
        \cellcolor[gray]{.9} $v_1$ & 1& 0.5 & 0.5\\ \hline
        \cellcolor[gray]{.9} $v_2$ & 3 & 2 & 1.5 \\ \hline
        \cellcolor[gray]{.9} $v_3$ & \underline{4} & 2.5 & 2 \\ \hline
        \cellcolor[gray]{.9} $v_4$ & \underline{5} & \underline{4} & \underline{3}  \\ \hline
        \cellcolor[gray]{.9} $t$ & 10 & 6.5 & \underline{\textcolor{red}{5.5}} \\ \hline
    \end{tabular}
  }%
\end{table}

\subsection{Proof of Correctness of Algorithm~\ref{alg:dijkstra_mod}}
We continue with a proof of correctness and running time of \Cref{alg:dijkstra_mod}
that solves \GraphStar{} for a single agent.

\OneG*
\begin{proof}
\label{pro:OneG}
We now prove correctness of the algorithm and analyze its running time.
Each updated distance is always at least as large as the distance to the pivot pair, so no already‐visited pair is ever updated, and pairs are extracted in ascending order of distance. Since the graph is connected, every routing pair will eventually be visited exactly once. We now show by induction that, when pivot pair $(v,b)$ is extracted, then  $D(v,b)$ denotes the minimum cost of a path from $s$ to $v$ reducing the cost of at most $b$ edges.

As a base case, the first pivot is $(s,0)$ and $D(s,0) = 0$ is clearly correct. 
Now assume that the first $i$ pivots $r_{1},\dots,r_{i}$ had the correct value $D(r_j)$ when extracted. 
Let $r_{i+1} = (v,b)$ be the next pivot pair with computed distance $d_{i+1} := D(v,b)$. Suppose that the lowest cost of a path from $s$ to $v$ investing in at most $b$ edges is $d^{*}$.

We first claim that $d_{i+1}\ge d^*$.
Let $j\in [i]$ be the index so that $r_j$ is the last pivot pair that led to an update of $r_{i+1}$.
Assume that $r_j = (v',b')$.
By induction, there is a path of cost $D(v',b')$ from $s$ to $v'$ investing into at most $b'$ edges.
If the last update was a budget-increasing update, then $v' = v$, $b' = b-1$, and the update was $D(v,b) = D(v',b')$.
Hence, there exists a path from $s$ to $v$ at cost $D(v,b)$ using at most $b$ edges, and we have $d_{i+1}\ge d^*$.
If it was a reducing update, then $b = b'+1$ and $D(v,b) = D(v',b') + \alpha w(v',v)$. 
If it was a non-reducing update, then $b = b'$ and $D(v,b) = D(v',b') + w(v',v)$. 
Consider adding the edge $\{v',v\}$ to the path to $v'$ as a reduced edge in the first case and non-reduced edge in the second case.
This yields a path from $s$ to $v$ of weight $D(v,b)$ using using at most $b$ edges, and again $d_{i+1}\ge d^*$.

Next, we prove $d^* \ge d_{i+1}$.
Let $P$ be a shortest path to $r_{i+1}$ with at most $b$ reductions, and let $(s,0),p_{1},\dots,p_{k}$ be the corresponding routing pairs when pairing the vertices along $P$ with the number of reduced edges up to this vertex.
Note that $p_k = (v,b')$ for some $b' \le b$.
If $b' < b$, we append the pivot pairs $(v,b'+1),\dots, (v,b)$.
Let $r = (\hat v,\hat b)$ be the first pivot pair in this sequence that had not been processed when $(v,b)$ was selected as a pivot pair.
Since $(v,b)$ is in the sequence and had not been processed, such a pivot pair exists.
Let $(\hat v',\hat b')$ be its predecessor.
Let $\hat d$ be the cost of a shortest path from $s$ to $\hat v$ using at most $\hat b$ edges.
By the minimality in the choice of $(v,b)$, we know that 
\begin{equation}\label{eq:pred}
    d_{i+1} = D(v,b)\le D(\hat v,\hat b)\text.
\end{equation}

If $\hat v' = \hat v$, then $r$ was appended at the end of $P$. 
Hence, when $(\hat v',\hat b')$ was processed and tested for a budget-increasing update, we have $D(\hat v,\hat b) = D(\hat v',\hat b')\le d^*$.

Now, assume that $\hat v' \neq \hat v$.
Note that, by induction, the cost of the shortest path from $s$ to $\hat v'$ with $\hat b'$ reduced edges is $D(\hat v',\hat b')$.
Now, the 
segment of $P$ up to $\hat v'$ is such a shortest path.
Otherwise, we could replace this part of $P$ by a shorter path to $\hat{v}'$ using a budget of $b'$. Since the length of the second part of $P$ does not change, this would yield a shorter path to $r$ with budget $b$, contradicting that $P$ is a shortest such path.
Hence, $D(\hat v',\hat b')$ is the cost of $P$ up to $\hat v'$.

If $\hat b = \hat b'$, then $P$ contains the edge $\{\hat v',\hat v\}$ and we have $d^* \ge D(\hat v',\hat b') + w(\hat v', \hat v)\ge D(\hat v,\hat b)$, where the latter inequality follows from a potential non-reducing update when $(\hat v',\hat b')$ was processed.
Moreover, if $\hat b \neq \hat b'$, then $\hat b = \hat b'+1$ and $P$ reduces edge $\{\hat v',\hat v\}$.
We then have $d^* \ge D(\hat v',\hat b') + \alpha w(\hat v', \hat v)\ge D(\hat v,\hat b)$, where the latter inequality follows from a potential reducing update when $(\hat v',\hat b')$ was processed.

In all three cases, we conclude that $d^* \ge D(\hat v,\hat b)$.
Combining this with \Cref{eq:pred}, we obtain $d^* \ge d_{i+1}$.

Next, we will show that for the solution $S$ returned by the algorithm, it holds that $\delta_S(s,t) = D[t,\beta]$, i.e., it yields the optimal cost.
To this end, we refer to the function $p[v,b]$ used in \Cref{alg:dijkstra_mod} as the \emph{predecessor function} and, for $k\in \mathbb N$ and $r = (v,b)$, we use $p^k(r)$ to refer to the pivot pair when this function was applied $k$ times.
Now, let $r_1 = (t,\beta)$ and let $\rho$ be the unique integer with $p^{\rho}(r) =(s,0)$, i.e., this is the number of pivot pairs traced by the predecessor function until reaching $(s,0)$.
For $i\in [\rho]$, we write $r_i = (v_i,b_i)$ and
let $S_i = S\setminus \{\{v_j,v_{j+1}\}\colon j\in [i-1], v_j\neq v_{j+1}, b_j = b_{j+1}+1\}$.
Clearly, $S = S_1$.

We now prove by downward induction for $i\in [\rho]$ that 
$\delta_{S_i}(s,v_i) = D(v_i,b_i)$.
For $i = \rho$, this is immediate since $r_{\rho} = (s,0)$ and $D(r_{\rho}) = 0$.

Now consider any $i\in [\rho - 1]$.
First, note that $|S_i|\le b_i$ as $S$ contains at most one element for each time when the budget is decreased.
Hence, $\delta_{S_i}(s,v_i)\ge D(r_i)$ by the first part of the proof.

It remains to prove the reverse inequality. 
By induction, there exists a path from $s$ to $v_{i+1}$ of cost $\delta_{S_{i+1}}(s,v_{i+1}) = D(v_{i+1},b_{i+1}) = D(r_{i+1})$.
We know that $p(r_{i+1}) = r_i$, i.e., $r_i$ was marked as $r_{i+1}$'s predecessor.
If this occurred in a budget-increasing update, then $v_i = v_{i+1}$ and $S_i = S_{i+1}$.
Hence, $\delta_{S_i}(s,v_i) = \delta_{S_{i+1}}(s,v_{i+1}) = D(r_{i+1}) = D(r_i)$, where the last equality followed from the update of the distance function when $r_i$ was marked as predecessor.

If $r_i$ was marked during a non-reducing update, then $v_i\neq v_{i+1}$ and $S_i = S_{i+1}$
Since we can add $\{v_i,v_{i+1}\}$ to any path to $v_{i+1}$, we have
$\delta_{S_i}(s,v_i) \le \delta_{S_i}(s,v_{i+1}) + w(v_i,v_{i+1}) = \delta_{S_{i+1}}(s,v_{i+1}) + w(v_i,v_{i+1}) = D(r_{i+1}) + w(v_i,v_{i+1}) = D(r_i)$.

Finally, if $r_i$ was marked during a reducing update, then $v_i\neq v_{i+1}$ and $S_i = S_{i+1}\cup\{\{v_i,v_{i+1}\}\}$.
Again, we can add $\{v_i,v_{i+1}\}$ to any path to $v_{i+1}$. 
Hence,
$\delta_{S_i}(s,v_i) \le \delta_{S_{i+1}}(s,v_{i+1}) + \alpha w(v_i,v_{i+1})= D(r_{i+1}) + \alpha w(v_i,v_{i+1}) = D(r_i)$.

This completes the proof of correctness of \Cref{alg:dijkstra_mod}, as $S = S^1$ yields a path of cost $D(t,\beta)$.

We proceed with the running time.
Recall that $T_\text{dk}$ and $T_\text{em}$ represent the time costs of the \emph{decrease key} and \emph{extract minimum} operations, respectively.

Assume that the underlying graph $G$ has $m$ vertices and $\ell$ edges.
Initializing $D$ takes $\mathcal{O}(\beta m)$ time. 
The main loop to determine all entries of $D$ performs 
$\bigO{\beta m}$ 
extractions and at most one budget‐increase per extraction, costing $\bigO{\beta m T_\text{em}}$. 
Each edge is considered twice per budget level, yielding $\bigO{\beta \ell T_\text{dk}}$ over all reducing and non‐reducing updates. 
Since $\beta\le m$ and $m\le \ell\le m^{2}$, the total is 
\[
\bigO{\beta m T_{\mathrm{em}}+\beta \ell T_{\mathrm{dk}}}\text.
\]
With a Fibonacci heap ($T_{\mathrm{em}}=\bigO{\log(\beta m)}$, $T_{\mathrm{dk}}=\bigO{1}$ amortized)~\citep{thomas2009introduction} this becomes 
\[
\bigO{\beta m\log m + \beta\ell}\text.
\]
Finally, determining the set $S$ in the final part of the algorithm requires at most $\bigO{\beta m}$ steps, which does not add significant running time.
\proofqed\end{proof}

\section{Fixed Number of Agents}\label{app:fixedN}

Next, we provide the proof for our result concerning a fixed number of~$2$ agents.

\TwoG*

\begin{proof}
As a preliminary remark, running Algorithm~\ref{alg:dijkstra_mod} yields, for each source–target pair, a budget mapping $\mu\colon[\beta]_{0}\to\mathbb{Q}_{0}^{+}$.  Given two such mappings $\mu_{1}$ and $\mu_{2}$, monotonicity in the budget implies that more budget never increases cost.  Hence for any binary operation $\circ\in\{+,\max\}$ one can compute $(\mu_{1}\circ\mu_{2})(\beta)=\min_{b\in[\beta]_{0}}\bigl(\mu_{1}(b)\,\circ\,\mu_{2}(\beta-b)\bigr)$ in $\mathcal{O}(\beta)$ time, and merging the two mappings into a new mapping of the same form takes $\mathcal{O}(\beta^{2})$ time.

By running the modified Dijkstra algorithm from every vertex $v$ to every vertex $u$, we obtain the mappings $\mu_{vu}$ for all $(v,u)$. 
An optimal solution to a 2-\GraphStar{} instance falls into one of two cases, depending on whether the two agents’ shortest paths intersect.
To obtain the optimal solution, we compute each of them and take the best.
Also, we compute the running time for each of these options. 
To obtain the total running time, we sum up the running time for all of these.

Let the agents be $(s,t)$ and $(s',t')$.
If their shortest paths are disjoint, then computing $(\mu_{st}+\mu_{s't'})(\beta)$ (utilitarian objective) or $(\max\{\mu_{st},\mu_{s't'}\})(\beta)$ (egalitarian objective) takes $\mathcal{O}(\beta)$.

Otherwise the intersection forms a connected sub-path from $p$ to $q$.
Then, one of the endpoints of this sub-path, say $p$ is either reached from $s$ and $s'$ or from $s$ and $t'$.
Since the computation of both cases, is totally symmetric, we focus on the former case.
Eventually, this only leads to a doubling in running time and the best solution where the agents' paths intersect is the one of lower cost among the two cases.

Assuming $p$ and $q$ are guessed, each agent’s path splits into three segments, so we first compute $\mu_{1}=\mu_{sp}+\mu_{qt}$ and $\mu_{2}=\mu_{s'p}+\mu_{qt'}$, which are the costs of reaching the two junctions where agents merge from their respective terminals.
Under the egalitarian objective we then set $\mu_{3}=\max\{\mu_{1},\mu_{2}\}$ and finally compute $\mu(\beta)=(\mu_{3}+\mu_{pq})(\beta)$; under the utilitarian objective we set $\mu_{3}=\mu_{1}+\mu_{2}$ and compute $\mu(\beta)=(\mu_{3}+2\,\mu_{pq})(\beta)$. 
For fixed $p$ and $q$, these merges cost $\mathcal{O}(\beta^{2})$, and accounting for the $m^2$ possibilities for selecting $p$ and $q$, we obtain $\mathcal{O}(m^2\beta^{2})$ for this step.

Putting everything together, the total running time consists of computing the budget mappings for all pairs of terminals and then merging them.
By \Cref{cor:budmap}, the former can be done in time $\mathcal{O}(\beta m^2\log m + m\beta \ell)$ so the total running time is $\mathcal{O}\bigl(\beta m^2\log m + m\beta \ell+\beta+m^{2}\beta^{2}\bigr) = \mathcal{O}\bigl(\beta m^2\log m + m\beta \ell+m^{2}\beta^{2}\bigr)$. 
\proofqed\end{proof}

Next, we prove the important structural lemma for our XP algorithm.

\PathConsistency*

\begin{proof}
    Consider a weighted graph $G = (V,E,w)$ and and a collection of $k$ terminal pairs $\{(s_i,t_i)\colon i\in [k]\}$, where $s_i,t_i\in V$.
    Let $$\epsilon := \min\{|w(F)-w(F')| \colon F, F'\subseteq E \text{ with } w(F) \neq w(F')\}\text.$$

    This is the smallest difference in weight of any two edge sets of different weight.
    Note that $\epsilon > 0$ if $G$ contains at least one edge $e$ with nonzero weight (take $F= \{e\}$ and $F' = \emptyset$), which we may assume without loss of generality.
    
    We define a perturbed weight function. 
    For this, enumerate edges arbitrarily, i.e., for $\ell = |E|$, let $E = \{e_1,\dots, e_\ell\}$.
    For $i\in [\ell]$, we set $w'(e_i) = w(e_i) + 2^{-i}\epsilon$.
    We next claim that shortest paths with respect to the perturbed weight functions are shortest paths with respect to the original weights.
    
    \begin{claim}\label{cl:SPpreserve}
        Let $u,v\in V$.
        If $P$ is a shortest $u$-$v$-path with respect to $w'$, then $P$ is a shortest $u$-$v$-path with respect to $w$.
    \end{claim}

    \begin{claimproof}
        Let $u,v\in V$ and let $P$ be a shortest $u$-$v$-path with respect to $w'$.
        Assume for contradiction that there exists $P'\in P(u,v)$ with $w(P') > w(P)$.
        By our definition of $\epsilon$, it holds that $w(P') - w(P) \ge \epsilon$.
        Hence,
        $$w'(P') - w'(P) = \underbrace{w'(P') - w(P')}_{> 0 } 
        + \underbrace{w(P') - w(P)}_{\ge \epsilon} 
        + \underbrace{w(P) - w'(P)}_{> -\epsilon } > 0\text.$$
        There, we use that the sum of all perturbations to all edges is strictly less than $\epsilon$.
        Hence, $P$ is no shortest $u$-$v$-path with respect to $w'$, a contradiction.
    \claimqed\end{claimproof}

    Second, we observe that shortest $u$-$v$-paths with respect to $w'$ are unique.

    \begin{claim}\label{cl:SPunique}
        Let $u,v\in V$.
        Then there exists a unique shortest $u$-$v$-path with respect to $w'$.
    \end{claim}

    \begin{claimproof}
        Let $u,v\in V$ and assume for contradiction that $P$ and $P'$ are two different shortest $u$-$v$-paths with respect to $w'$.
        Let $D = E(P)\Delta E(P')$ be the symmetric difference of edges between $P$ and $P'$, i.e., the set of edges contained in exactly one of the two paths.
        Let $i^* = \min\{i\in [\ell]\colon e_i\in D\}$, i.e., the smallest index of an edge in the symmetric difference.
        Without loss of generality, we assume that $e_{i^*}\in E(P)\setminus E(P')$.
        By \Cref{cl:SPpreserve}, it holds that $P$ and $P'$ are shortest $u$-$v$-paths with respect to $w$. 
        Hence $w(P) - w(P') = 0$
        As a result,
        \begin{align*}
            w'(P) - w'(P') &= \left(w(P) + \sum_{i\in [\ell]\colon e_i\in E(P)} 2^{-i}\epsilon\right) - \left(w(P') + \sum_{i\in [\ell]\colon e_i\in E(P')} 2^{-i}\epsilon\right) \\
            & = \epsilon\left(\sum_{i\in [\ell]\colon e_i\in E(P)\setminus E(P')} 2^{-i} - \sum_{i\in [\ell]\colon e_i\in E(P')\setminus E(P)} 2^{-i}\right)\\
            & > \epsilon \left(2^{-i^*} - \sum_{i\ge i^* + 1}^{\infty} 2^{-i}\right) = 0\text.
        \end{align*}
        This contradicts that $P$ and $P'$ have the same weight with respect to $w'$.
    \claimqed\end{claimproof}

    We are ready to conclude the proof.
    For every $i\in [k]$ let $P_i$ be a shortest $s_i$-$t_i$-path with respect to $w'$.
    By \Cref{cl:SPpreserve}, $P_i$ is a shortest $s_i$-$t_i$-path with respect to $w$.
    Hence, our first property is satisfied.
    Moreover, by \Cref{cl:SPunique}, the shortest $s_i$-$t_i$-paths with respect to $w'$ are unique.
    By \Cref{lem:UniqueConSP}, the collection $\{P_i\colon i\in [k]\}$ is consistent.
\proofqed\end{proof}

Next, we complete the proof of \Cref{thm:XPfixedK} by computing the running time of \Cref{alg:XPagents}.

\XPfixedK*

\begin{proof}[Proof of running time]
    We proceed with estimating the running time.
    First, by \Cref{cor:budmap}, we can obtain all budget mappings in time $\bigO{\beta m^2\log m + m\beta \ell} \le \bigO{\beta m^3}$.
    In the same time, we can extract the final solution in \Cref{ln:ExtractSol}.
    
    Next, we bound the number of (combinations of) guesses in lines~\ref{ln:junctions}, \ref{ln:internaljunctions}, and \ref{ln:BudgetToTravel} of \Cref{alg:XPagents}.
    There are at most $\binom{m}{k^2}\le m^{k^2}$ possibilities for $J$.
    For every $J$ and every agent, there are at most 
    $$ 1 + \sum_{i = 1}^{2k} \prod_{j = 1}^i (2k^2 - j + 1) \le \sum_{i = 0}^{2k}\left(2k^2\right)^i\le 2 \left(2k^2\right)^{2k} \le (2k)^{4k}$$
    guesses for the junction vertices of an agent.
    There the first sum consists of the one case of having no junction vertices plus the cases for choosing an ordered tuple of up to $2k$ junction vertices for a set of $2k^2$ elements.
    Since there are $k$ agents, we have a total of at most $(2k)^{4k^2}$ possibilities to choose the junction vertices corresponding to paths.

    Finally, we allocate a budget of $\beta$ to a set of at most $(2k+1)\cdot k$ meta edges.
    The number of possibilities for this can be bounded by the number of possibilities to draw an unordered sample of $\beta$ elements from a set of $(2k+1)\cdot k$ elements.
    For this there are $\binom{(2k+1)\cdot k + \beta - 1}{\beta}$ many possibilities.
    To bound this, observe that for $\beta \ge 1$, it holds that 
    \begin{equation*}
        \prod_{i = 1}^{N-1} (\beta + i) \le \prod_{i = 1}^{N-1} \beta(i + 1) = \beta^{N-1}N!
    \end{equation*}
    Using this for $N = (2k+1)\cdot k$, we obtain
    \begin{align*}
        &\phantom{=}\binom{(2k+1)\cdot k + \beta - 1}{\beta} 
        = \frac{[(2k+1)\cdot k + \beta - 1]!}{[(2k+1)\cdot k - 1]!\beta!}
         = \frac{\prod_{i = 1}^{(2k+1)\cdot k - 1} (\beta + i)}{[(2k+1)\cdot k - 1]!}\\
         &\le \frac{\beta^{(2k+1)\cdot k - 1} [(2k+1)\cdot k]!}{[(2k+1)\cdot k - 1]!}
         = (2k+1) k \beta^{(2k+1)\cdot k - 1} \le (2k)^2\beta^{3k^2}\text.
    \end{align*}

    Hence, in total, we have to consider at most
    \begin{equation*}
        m^{k^2}  (2k)^{4k^2}(2k)^2 \beta^{3k^2} = (2k)^{4k^2+2}m^{k^2}\beta^{3k^2}
    \end{equation*}
    many combinations of guesses.
    For each of them, we have to perform at most $2k+1$ look-ups in our budget maps for each agent, consuming time \bigO{k^2}. We then aggregate these values to a utilitarian or egalitarian cost in time \bigO{k^2}.
    Finally, we have to find the minimum value among all guesses. For this, we can always store the current minimum value, and then compare (and potentially replace) it whenever we have a value for a new guess. This only adds constant time to the operations for each guess.

    In total we obtain a running time of $$\bigO{\beta m^3 + (2k)^{4k^2+2}m^{k^2}\beta^{3k^2}k^2} = \bigO{(2k)^{4k^2+4}m^{k^2}\beta^{3k^2}}\text.$$
    This completes the proof.
\proofqed\end{proof}

Next, we prove our W[1]-hardness result.

\WoneGraph*

\begin{proof}
    Let $\alpha\in [0,1)$.
    We provide a parameterized reduction from \MCClique.
    An instance consists of a tuple $\langle G, k\rangle$, where $G = (V,E)$ is a graph together with a proper coloring $(V_1,\dots, V_k)$ of the vertices $V$, and $k$ is a positive integer.
    An instance is a Yes-instance if there exists a clique of size $k$ with one vertex in each color class.
    \MCClique{} is known to be W[1]-hard parameterized by $k$ \citep{FHRV09a}.

    Given an instance $\langle G, k\rangle$ of \MCClique{} with graph $G = (V,E)$ and vertex coloring $(V_1,\dots, V_k)$, we define a reduced instance $\bic = \langle G', \mathcal{A}, \alpha, \beta \rangle$ of \Graph{} as follows:

    Let $G' = (V',E',w)$ be the underlying graph defined by 
    \begin{itemize}
        \item $V' = V \cup \{v_c\colon c\in [k]\}$, i.e., we copy the original vertices and a vertex for each color class,
        \item $E' = E \cup \bigcup_{c\in [k]}\{\{v_c,w\}\colon w\in V_c\}$, i.e., we copy the original edges and edges from the dedicated color class vertices to all vertices in $V$ of their color class, and 
        \item $w(e) = \begin{cases}
            1 &\text{if }e\in E,\\
            k &\text{else.}
        \end{cases}$, i.e., original edges have weight~$1$ and the newly added edges connecting color vertices to their color class have a weight of~$k$. 
    \end{itemize}

    We have one agent for each pair of color classes with terminals corresponding to their respective color class vertices, i.e., $\mathcal A = \{\{(v_c,v_{c'})\colon 1\le c < c'\le k\}\}$.
    Finally, we set $\beta = k$.
    Note that the number of agents is $\binom{k}{2}$, and hence it is a computable function with respect to $k$.
    Similarly, the budget of $k$ is a computable function with respect to $k$.
    Hence, the reduction is a valid parameterized reduction.

    It remains to prove correctness.
    We claim that the following are equivalent:
    \begin{enumerate}
        \item The original instance $\langle G, k\rangle$ contains a clique with one vertex in every color class.\label{it:OrigClique}
        \item The reduced instance contains a feasible solution $F\subseteq E'$ with $c\egsup(F) \le 2k\alpha +1$.\label{it:EgalCost}
        \item The reduced instance contains a feasible solution $F\subseteq E'$ with $c\utsup(F) \le (2k\alpha + 1) \binom{k}{2}$.\label{it:UtilCost}
    \end{enumerate}

    Clearly, (\ref{it:EgalCost}) implies (\ref{it:UtilCost}).
    We proceed with the remaining implications.

    \noindent\textbf{(\ref{it:OrigClique}) $\Rightarrow$ (\ref{it:EgalCost})} 
    Assume that $G$ contains a clique $C$ with one vertex in every color class.
    For $c\in [k]$, denote by $w_c$ the vertex in $C\cap V_c$.
    Consider $F = \{\{v_c,w_c\}\colon c\in [k]\}$.
    Clearly, $|F|= k$, so $F$ is feasible.
    
    Now, let $1\le c < c' \le k$ and consider the agent $a = (v_c,v_{c'})$.
    Then, $\delta_F(a) \le 2k\alpha +1$, as the agent can travel from $v_c$ to $w_c$ for a discounted cost of $k\alpha$, then to $w_{c'}$ for a cost of~$1$ (this edge exists as $C$ is a clique), and then from $w_{c'}$ to $v_{c'}$ for a discounted cost of $k\alpha$.
    Hence, $c\egsup(F) \le 2k\alpha + 1$.

    \noindent\textbf{(\ref{it:UtilCost}) $\Rightarrow$ (\ref{it:OrigClique})}
    Let $F\subseteq E'$ be a feasible solution with $c\utsup(F) \le (2k\alpha + 1) \binom{k}{2}$.
    We first claim that for every $c\in [k]$ there exists an edge $f\in F$ with $v_c\in f$.
    Assume for contradiction that for some $c\in [k]$ no such edge is in $F$.
    Then, each of the $k$ agents having $v_c$ as a terminal has to use a nondiscounted edge from $v_c$ and then has to use at least one edge in $E$ and one edge incident to their second terminal.
    Hence, if $a$ is such an agent, then $\delta_F (a)\ge k + (1+k)\alpha$.
    Moreover, for any agent $a'$ it clearly holds that $\delta_F(a') \ge (1+2k)\alpha$.
    Thus,

    \begin{align*}
        c\utsup(F) &\ge k \left[ k + (1+k)\alpha\right] + \left[\binom{k}{2} - k\right](1+2k)\alpha\\
        & = k[(2k\alpha+1) + (1-\alpha)(k-1)]  + \left[\binom{k}{2} - k\right][(2k\alpha+1) - (1-\alpha)]\\
        & = (2k\alpha+1)\binom{k}{2} + (1-\alpha)\left[k(k-1) - \left[\binom{k}{2} - k\right]\right] \\
        & = (2k\alpha+1)\binom{k}{2} + (1-\alpha)\left[\frac 12\binom{k}{2} + k\right]\\
        & > (2k\alpha+1)\binom{k}{2}\text.
    \end{align*}

    The strict inequality follows because $\alpha\in [0,1)$. 
    However, this contradicts that $c\utsup(F)\le (2k\alpha+1)\binom{k}{2}$.
    Hence, for every $c\in [k]$ there indeed exists an edge $f\in F$ with $v_c\in f$.
    Let $w_c\in V_c$ such that $f = \{v_c,w_c\}$, i.e., $w_c$ is the other endpoint of the edge in $F$ incident to $v_c$.

    Consider $C = \{v_c\colon c\in [k]\}$.
    Clearly, $C$ contains one vertex of each color class.
    We claim that $C$ forms a clique in $G$.

    Since there are $k$ edges in $F$ incident to color class vertices, the feasibility of $F$ implies that $F\cap E = \emptyset$, i.e., $F$ contains none of the edges corresponding to the ones of the source instance.
    Let $1\le c < c' \le k$ and consider the agent $a = (c,c')$.
    If $a$ does not travel along $\{v_c,w_c\}$ or $\{v_{c'},w_{c'}\}$, then her travel cost is at least $k + 1 + k\alpha > 2k\alpha + 1$.
    If they use both these edges, then their cost is $2k\alpha + 1$ if $\{w_c,w_{c'}\}\in E$ and strictly higher if $\{w_c,w_{c'}\}\notin E$.
    Hence, for each agent, regardless of their route, their minimum travel cost is $2k\alpha + 1$ if $\{w_c,w_{c'}\}\in E$ and strictly higher if $\{w_c,w_{c'}\}\notin E$.
    We conclude that, $c\utsup(F) = (2k\alpha + 1) \binom{k}{2}$ if $C$ is a clique and strictly higher otherwise.
    Since we assumed that $c\utsup(F) \le (2k\alpha + 1) \binom{k}{2}$, it must be the case that $C$ forms a clique.
\proofqed\end{proof}

\HubTree*

\begin{proof}
    Let everything be given as in the statement of the lemma and assume that $\{P_i\colon i\in [k]\}$ is consistent.
    Since the paths connect all vertices with $t$, we know that $(V_T,E_T)$ is connected.
    Assume for contradiction that $(V_T,E_T)$ contains a cycle $C$
    and let $i\in [k]$ so that the number of edges shared by $C$ and $P_i$ is maximized.
    Traversing $P_i$ from $s_i$ to $t$, let $v_1$ be the first vertex of $P_i$ that is on $C$ and $v_2$ be the last vertex of $P_i$ that is on $C$.
    Then, as $v_1$ is the first vertex of $P_i$ on $C$, at most one of the edges on $C$ incident to $v_1$ lies on $P_i$.
    Let $e_1$ be an edge on $C$ incident to $v_1$ that is not on $P_i$ and let $v_1'$ such that $e_1 = \{v_1,v_1'\}$.
    Let $j\in [k]$ with $e_1\in E(P_j)$.
    We will now derive a contradiction by making a case distinction based on the direction of traversal of $e_1$ by $P_j$.
    
    Assume first that $P_j$ traverses $e_1$ from $v_1$ to $v_1'$.
    Then $P_j$ shares with $P_i$ a segment that ends at $v_1$ (as $P_i$ does not traverse $e_1$. 
    However, as both paths terminate at $t$, they must share another segment later on.
    This contradicts consistency.

    Second, assume that $P_j$ traverses $e_1$ from $v_1'$ to $v_1$.
    Then, a segment of $P_j$ shared with $P_i$ starts at $v_1$. 
    This segment has to end before $P_1$ reaches $v_2$.
    Otherwise, $P_j$ traverses all edges on $C$ traversed by $P_i$ as well as $e_1$, contradicting the maximality assumption on $P_i$.
    However, as in the first case, there must be second segment shared by $P_j$ and $P_i$ ending in $t$. 
    Hence, we also reach a contradiction with consistency.

    We conclude that $(V_T,E_T)$ contains no cycle and, therefore, is a tree.
\proofqed\end{proof}

\FPThub*

\begin{proof}
    Consider an instance $\bic = \gbpins$ of \GraphHub{} with $G = (V,E,w)$ and hub $t\in V$, i.e., $\ac \sqsubseteq V\times \{t\}$.
    For an agent $a = (s,t)\in \ac$, we refer to $s$ as their \emph{individual terminal}.
    Given an agent $a$, we denote their individual terminal by $s_a$.
    We will solve \GraphHubStar{} via dynamic programming.
    As a preprocessing step, we run \Cref{alg:dijkstra_mod} to obtain budget mappings $\mu_{u,v}\colon [\beta]_0\to \mathbb Q_0^+$ for every pair of vertices $u,v\in V$ with $u\neq v$. 
    Moreover, for each $u\in V$, we define $\mu_{u,u}\colon [\beta]_0\to \mathbb Q_0^+$ by $\mu_{u,u}(b) = 0$ for all $b\in [\beta]_0$. 
        
    We now define the table for our dynamic program. 
    For every subset of agents $A\subseteq \ac$ with $A\neq \emptyset$, vertex $v\in V$, and budget $b\in [\beta]_0$, let $\Dut[A,v,b]$ (resp., $\Deg[A,v,b]$) denote the minimum possible utilitarian (resp., egalitarian) travel cost when all agents in $A$ travel from their respective individual terminals to $v$ using at most $b$ upgraded edges.
    As usual, we use the notation $\Dstar[A,v,b]$ to refer to either of the tables.
    It is easy to fill the table for sets $A$ consisting of single agents as these values are captured in our budget mappings obtained during our preprocessing.
    Specifically, for $a\in A$, $v\in V$, and $b\in [\beta]_0$, we set 
    \begin{equation}\label{eq:init}
        \Dstar[\{a\},v,b] = \mu_{s_a,v}(b)\text.
    \end{equation}

    Now, assume we are given an arbitrary tuple $(A,v,b)$ where $A\subseteq \ac$ with $|A|\ge 2$, $v\in V$, and $b\in [\beta]_0$.
    The crucial insight is how to recursively fill the table for larger agent sets.
    The idea is that, like in the proof of \Cref{thm:XPfixedK}, we consider optimal solutions where agents' paths to $v$ are consistent.
    By \Cref{lem:HubTree}, we know that there union forms a tree.
    Now, suppose we root the tree at $v$.
    We now trace the tree until we find a closest descendant $u$ of $v$ that is either an individual terminal of an agent in $A$ or has degree at least $2$ (it is possible that $u=v$.
    If $u$ is an individual terminal of agent $a\in A$, then the cost of all agents in $A$ traveling to $v$ splits into $A\setminus \{a\}$ traveling to $s_a$ and from there all agents travel to $v$ together.
    If $u$ is no individual terminal but of degree at least two, then we split $A$ into two nonempty subsets and consider their respective travel to $u$, and all agents traveling from $u$ to $v$ together.
    We obtain the update by considering an optimal cost split for the individual parts of travel.
    We capture the update in a claim.
    Since the update formulas differ based on a utilitarian and egalitarian cost aggregation, we provide separate claims for both cases.
    We start with the utilitarian case.

    \begin{claim}
        Let $A\subseteq \ac$ with $|A|\ge 2$, $v\in V$, and $b\in [\beta]_0$.
        Define 

        {\small
        \begin{align*}
            \Dut_1[A,v,b] &:= \min_{a\in A}\min_{b'\in[b]_0} \bigg[\Dut[A\setminus\{a\},s_a,b] + |A|\cdot\mu_{s_a,v}(b - b')\bigg]\\
            \Dut_2[A,v,b] &:= \min_{\emptyset \neq A'\subsetneq A} \min_{u\in V}\min_{b'\in[b]_0}\min_{b''\in [b']_0} \bigg[\Dut[A',u,b''] + \Dut[A\setminus A',u, b'-b''] + |A|\cdot \mu_{u,v}(b - b')\bigg]\\
        \end{align*}
        }
        
        Then, $\Dut[A,v,b] = \min \{\Dut_1[A,v,b], \Dut_2[A,v,b]\}$.
    \end{claim}
    \begin{claimproof}
        Let $A\subseteq \ac$ with $|A|\ge 2$, $v\in V$, and $b\in [\beta]_0$.
        We will first prove that $\Dut[A,v,b] \le \Dut_1[A,v,b]$.
        Therefore, let $a\in A$ and $b'\in [b]_0$.
        Then the cost for all agents traveling from their individual terminals to $v$ when upgrading $b$ edges is at most the cost of all agents in $A\setminus \{a\}$ traveling to $s_a$ using at most $b'$ upgraded edges plus $|A|$ times the cost for a single agent traveling from $s_a$ to $v$ using the remaining budget of $b-b'$ for upgrading edges.
        Hence, $\Dut[A,v,b]\le \Dut[A\setminus\{a\},s_a,b] + |A|\cdot\mu_{s_a,v}(b - b')$.
        Taking a minimum over all $a\in A$ and $b'\in [b]_0$ implies $\Dut[A,v,b] \le \Dut_1[A,v,b]$.

        Next, we prove $\Dut[A,v,b] \le \Dut_2[A,v,b]$.
        Therefore, consider any split of $A$ into $A'$ and $A\setminus A'$ where $\emptyset \neq A'\subsetneq A$.
        Consider any vertex $u\in V$ and budgets $b'\in [b]_0$, $b''\in [b']_0$.
        We now consider the cost when the agents in $A'$ and $A\setminus A'$ each travel to $u$ with their own budget, and then all agents travel from $u$ to $v$ together with the remaining budget.
        Therefore, we split the budget $b$ into three parts according to $b'$ and $b''$, and allocate them to the part traveled only by $A'$, the one only by $A\setminus A'$, and the part together.
        Since this constitutes one way of all agents in $A$ traveling from their individual terminals to $v$, the cost for all agents in $A$ traveling to $v$ with budget $b$ optimally is bounded by the sum of these costs, i.e.,
        $\Dut[A,v,b]\le \Dut[A',u,b''] + \Dut[A\setminus A',u, b'-b''] + |A|\cdot \mu_{u,v}(b - b')$.
        By taking the minimum over all possibilities, we obtain $\Dut[A,v,b] \le \Dut_2[A,v,b]$.
        Together, with the first inequality, we have shown that $\Dut[A,v,b] \le \min \{\Dut_1[A,v,b], \Dut_2[A,v,b]\}$.

        To show the reverse inequality, we consider a set $F\subseteq E$ of at most $b$ edges such that the total travel time of all agents in $A$ from their individual terminals to $v$ when upgrading at most $b$ edges is minimized.
        Let $w_F\colon E \to \mathbb Q_0^+$ with $$w_F(e) = \begin{cases}
        \alpha \cdot w(e) & e\in F\\
        w(e) & e\in E\setminus F
        \end{cases}\text.$$
        By \Cref{lem:PathConsistency}, there exists a consistent collection of paths $\{P_a\colon a\in A\}$ such that for each $a\in A$, $P_a$ is a shortest $s_a$-$t$-path on $(V,E,w_F)$.
        Since they are all shortest paths and $F$ minimizes total travel time, we have that 
        \begin{equation}\label{eq:PathsCost:ut}
            \Dut[A,v,b] = \sum_{a\in A} w_F(P_a)\text.
        \end{equation}
        
        By \Cref{lem:HubTree}, we know that the union of these paths forms a tree $T$.
        We root the tree at $t$ and let $u$ be the first descendant of $t$ that is either an individual terminal or has degree at least $2$.
        Clearly, such a vertex exists because $T$ contains all individual terminals.

        Assume first that $u$ is an individual terminal, say $u = s_a$ for $a\in A$, and let $b'\in [b_0]$ be such that $b-b'$ is the number of edges in $F$ on the unique path on $T$ from $s_a$ to $t$.
        Since the paths use $b - b'$ budget to travel from $s_a$ to $t$ and at most $b'$ budget to reach $s_a$ for all agents in $A\setminus \{a\}$, the weight of all paths is at least $\Dut[A\setminus\{a\},s_a,b] + |A|\cdot\mu_{s_a,v}(b - b')$.
        Combining this with \Cref{eq:PathsCost:ut}, we obtain
        \begin{align*}
            \Dut[A,v,b] &\ge \Dut[A\setminus\{a\},s_a,b] + |A|\cdot\mu_{s_a,v}(b - b')\\
            &\ge \Dut_1[A,v,b] \ge \min \{\Dut_1[A,v,b], \Dut_2[A,v,b]\}\text.
        \end{align*}

        Now assume that $u$ is a vertex of degree at least $2$ on $T$ and let $e = \{u,u'\}$ be an edge of $T$ such that $u'$ is a descendant of $u$.
        Let $T'$ be the subtree of $T$ containing $u'$ and all of $u'$s descendants.
        Let $A' = \{a\in A\colon s_a\in V(T')\}$, i.e., $A'$ is the subset of agents whose individual terminals are in $T'$.
        First, we observe that $A'\neq \emptyset$ as every agent $a$ such that $u'\in V(P_a)$, i.e., every agent whose path passes through $u'$, must have their individual terminal in $T'$.
        Similarly, as there must be agents with individual terminals in the subtrees of other descendants of $u$, we know that $A'\neq A$.
        Now let $b'\in [b]_0$ be such that $b-b'$ is the number of edges in $F$ on the unique path on $T$ from $u$ to $t$.
        Hence, we have at most $b'$ edges on $F$ on the subtree of $T$ containing $u$ and all its descendants.
        Let $b''\in [b']_0$ be the number of upgraded edges in $E(T')\cup \{e\}$.
        
        Now, by the choice of $u$, there is no individual terminal on $T$ before $u$, and hence all agents paths include a sub-path from $u$ to $t$ with $b-b'$ upgraded edges.
        Moreover, the agents in $A'$ and $A\setminus A'$ travel from their individual terminals to $u$ with a budget of $b''$ and at most $b'-b''$, respectively.
        Hence, the weight of all paths is at least $\Dut[A',u,b''] + \Dut[A\setminus A',u, b'-b''] + |A|\cdot \mu_{u,v}(b - b')$.
        Inserting this into \Cref{eq:PathsCost}, we obtain
        \begin{align*}
            \Dut[A,v,b] &\ge \Dut[A',u,b''] + \Dut[A\setminus A',u, b'-b''] + |A|\cdot \mu_{u,v}(b - b') \\
            &\ge \Dut_2[A,v,b] \ge \min \{\Dut_1[A,v,b], \Dut_2[A,v,b]\}\text.
        \end{align*}

        Together, we have shown that $\Dut[A,v,b] = \min \{\Dut_1[A,v,b], \Dut_2[A,v,b]\}$, which concludes the proof of the claim.
    \claimqed\end{claimproof}

    Next, we consider the recursion for an egalitarian total cost and provide the analogous proof.

    \begin{claim}
        Let $A\subseteq \ac$ with $|A|\ge 2$, $v\in V$, and $b\in [\beta]_0$.
        Define 
        {\small
        \begin{align*}
            \Deg_1[A,v,b] &:= \min_{a\in A}\min_{b'\in[b]_0} \bigg[\Deg[A\setminus\{a\},s_a,b] + \mu_{s_a,v}(b - b')\bigg]\\
            \Deg_2[A,v,b] &:= \min_{\emptyset \neq A'\subsetneq A} \min_{u\in V}\min_{b'\in[b]_0}\min_{b''\in [b']_0} \bigg[\max\bigg\{\Deg[A',u,b''], \Deg[A\setminus A',u, b'-b'']\bigg\} + \mu_{u,v}(b - b')\bigg]\\
        \end{align*}
        }

        Then, $\Deg[A,v,b] = \min \{\Deg_1[A,v,b], \Deg_2[A,v,b]\}$.
    \end{claim}
    \begin{claimproof}
        Let $A\subseteq \ac$ with $|A|\ge 2$, $v\in V$, and $b\in [\beta]_0$.
        We will first prove that $\Deg[A,v,b] \le \Deg_1[A,v,b]$.
        Therefore, let $a\in A$ and $b'\in [b]_0$.
        Then the egalitarian cost for all agents traveling from their individual terminals to $v$ when upgrading $b$ edges is at most the maximum cost among all agents in $A\setminus \{a\}$ traveling to $s_a$ using at most $b'$ upgraded edges plus the cost for a single agent traveling from $s_a$ to $v$ using the remaining budget of $b-b'$ for upgrading edges.
        Hence, $\Deg[A,v,b]\le \Deg[A\setminus\{a\},s_a,b] + 
        \mu_{s_a,v}(b - b')$.
        Taking a minimum over all $a\in A$ and $b'\in [b]_0$ implies $\Deg[A,v,b] \le \Deg_1[A,v,b]$.

        Next, we prove $\Deg[A,v,b] \le \Deg_2[A,v,b]$.
        Therefore, consider any split of $A$ into $A'$ and $A\setminus A'$ where $\emptyset \neq A'\subsetneq A$.
        Consider any vertex $u\in V$ and budgets $b'\in [b]_0$, $b''\in [b']_0$.
        We now consider the egalitarian cost when the agents in $A'$ and $A\setminus A'$ each travel to $u$ with their own budget, and then all agents travel from $u$ to $v$ together with the remaining budget.
        Therefore, we split the budget $b$ into three parts according to $b'$ and $b''$, and allocate them to the part traveled only by $A'$, the one only by $A\setminus A'$, and the part together.
        Since this constitutes one way of all agents in $A$ traveling from their individual terminals to $v$, the egalitarian cost for all agents in $A$ traveling to $v$ with budget $b$ optimally is bounded by the maximum of these costs, i.e.,
        $$\Deg[A,v,b]\le \max\bigg\{\Deg[A',u,b''],\Deg[A\setminus A',u, b'-b'']\bigg\} + \mu_{u,v}(b - b')\text.$$
        By taking the minimum over all possibilities, we obtain $\Deg[A,v,b] \le \Deg_2[A,v,b]$.
        Together, with the first inequality, we have shown that $\Deg[A,v,b] \le \min \{\Deg_1[A,v,b], \Deg_2[A,v,b]\}$.

        To show the reverse inequality, we consider a set $F\subseteq E$ of at most $b$ edges such that the maximum travel time of all agents in $A$ from their individual terminals to $v$ when upgrading at most $b$ edges is minimized.
        Let $w_F\colon E \to \mathbb Q_0^+$ with $$w_F(e) = \begin{cases}
        \alpha \cdot w(e) & e\in F\\
        w(e) & e\in E\setminus F
        \end{cases}\text.$$
        By \Cref{lem:PathConsistency}, there exists a consistent collection of paths $\{P_a\colon a\in A\}$ such that for each $a\in A$, $P_a$ is a shortest $s_a$-$t$-path on $(V,E,w_F)$.
        Since they are all shortest paths and $F$ minimizes maximum travel time, we have that 
        \begin{equation}\label{eq:PathsCost}
            \Deg[A,v,b] = \max_{a\in A} w_F(P_a)\text.
        \end{equation}
        
        By \Cref{lem:HubTree}, we know that the union of these paths forms a tree $T$.
        We root the tree at $t$ and let $u$ be the first descendant of $t$ that is either an individual terminal or has degree at least $2$.
        Clearly, such a vertex exists because $T$ contains all individual terminals.

        Assume first that $u$ is an individual terminal, say $u = s_a$ for $a\in A$, and let $b'\in [b_0]$ be such that $b-b'$ is the number of edges in $F$ on the unique path on $T$ from $s_a$ to $t$.
        Since the paths use $b - b'$ budget to travel from $s_a$ to $t$ and at most $b'$ budget to reach $s_a$ for all agents in $A\setminus \{a\}$, the maximum weight among all paths is at least $\Deg[A\setminus\{a\},s_a,b] + \mu_{s_a,v}(b - b')$.
        Combining this with \Cref{eq:PathsCost}, we obtain
        \begin{align*}
            \Deg[A,v,b] &\ge \Deg[A\setminus\{a\},s_a,b] + \mu_{s_a,v}(b - b')\\
            &\ge \Deg_1[A,v,b] \ge \min \{\Deg_1[A,v,b], \Deg_2[A,v,b]\}\text.
        \end{align*}

        Now assume that $u$ is a vertex of degree at least $2$ on $T$ and let $e = \{u,u'\}$ be an edge of $T$ such that $u'$ is a descendant of $u$.
        Let $T'$ be the subtree of $T$ containing $u'$ and all of $u'$s descendants.
        Let $A' = \{a\in A\colon s_a\in V(T')\}$, i.e., $A'$ is the subset of agents whose individual terminals are in $T'$.
        First, we observe that $A'\neq \emptyset$ as every agent $a$ such that $u'\in V(P_a)$, i.e., every agent whose path passes through $u'$, must have their individual terminal in $T'$.
        Similarly, as there must be agents with individual terminals in the subtrees of other descendants of $u$, we know that $A'\neq A$.
        Now let $b'\in [b]_0$ be such that $b-b'$ is the number of edges in $F$ on the unique path on $T$ from $u$ to $t$.
        Hence, we have at most $b'$ edges on $F$ on the subtree of $T$ containing $u$ and all its descendants.
        Let $b''\in [b']_0$ be the number of upgraded edges in $E(T')\cup \{e\}$.
        
        Now, by the choice of $u$, there is no individual terminal on $T$ before $u$, and hence all agents paths include a sub-path from $u$ to $t$ with $b-b'$ upgraded edges.
        Moreover, the agents in $A'$ and $A\setminus A'$ travel from their individual terminals to $u$ with a budget of $b''$ and at most $b'-b''$, respectively.
        Hence, the maximum weight among all all paths is at least 
        $$\max \bigg\{\Deg[A',u,b''], \Deg[A\setminus A',u, b'-b'']\bigg\} + \mu_{u,v}(b - b')\text.$$
        Inserting this into \Cref{eq:PathsCost}, we obtain
        \begin{align*}
            \Deg[A,v,b] &\ge \max \bigg\{\Deg[A',u,b''], \Deg[A\setminus A',u, b'-b'']\bigg\} + \mu_{u,v}(b - b') \\
            &\ge \Deg_2[A,v,b] \ge \min \{\Deg_1[A,v,b], \Deg_2[A,v,b]\}\text.
        \end{align*}

        Together, we have shown that $\Deg[A,v,b] = \min \{\Deg_1[A,v,b], \Deg_2[A,v,b]\}$, which concludes the proof of the claim.
    \claimqed\end{claimproof}

    Since, we have shown correctness of the recursion, we can solve \GraphHub{} by filling the table entries $\Dstar(A,v,b)$ for increasing cardinality of $A$. 
    Then, the correct solution is captured in $\Dstar(A,t,\beta)$.

    It remains to compute the running time of this procedure for instances with $k$ agents.
    By \Cref{cor:budmap}, the precomputation to obtain all budget mappings can be performed in time $\mathcal O(\beta m^2\log m + m\beta \ell) = \mathcal O(m^3 \beta)$.
    It remains to determine the time for filling the table.

    We first initialize the table as described in \Cref{eq:init} for the singleton agent sets in time $k \mathcal O(m\beta)$.

    Now, we observe that it is rather costly to directly determine the entries $\Dstar_2(A,v,b)$.
    However, we can avoid some redundant computation caused by the nested computation of the minimum.
    For this, we create an intermediate table $\Mstar(A,u,b')$ where $A\subseteq \ac$ with $|A|\ge 2$, $u\in V$, and $b'\in [\beta]_0$ defined by 
    $$\Mut(A,u,b') = \min_{\emptyset \neq A'\subsetneq A} \min_{b''\in [b']_0} [\Dut[A',u,b''] + \Dut[A\setminus A',u, b'-b'']]$$
    in case of utilitarian cost and  
    $$\Meg(A,u,b') = \min_{\emptyset \neq A'\subsetneq A} \min_{b''\in [b']_0} \max\{\Deg[A',u,b''], \Deg[A\setminus A',u, b'-b'']\}$$
    in case of egalitarian cost.
    
    This table extracts the part of the computation corresponding to merging two subsets of agents at a node $u$.
    Hence, whenever we want to determine $\Dstar(A,v,b)$, we first compute $\Dstar_1(A,v,b)$ (Step~1) and $\Mstar(A,u,b')$ (Step~2) and subsequently $\Dstar_2(A,v,b)$ (Step~3) as

    $$\Dut_2[A,v,b] = \min_{u\in V}\min_{b'\in[b]_0} \bigg[\Mut(A,u,b') + |A|\cdot \mu_{u,v}(b - b')\bigg]$$
    or 
    $$\Deg_2[A,v,b] = \min_{u\in V}\min_{b'\in[b]_0} \bigg[\Mut(A,u,b') + \mu_{u,v}(b - b')\bigg]\text.$$

    Finally, we compute $\Dstar(A,v,b) = \min \{\Dstar_1(A,v,b), \Dstar_2(A,v,b)\}$ (Step~4).

    Step~1, i.e., computing $\Dstar_1(A,v,b)$ for any $A\subseteq \ac$ with $|A|\ge 2$, $v\in V$, and $b\in [\beta]_0$, takes time at most $2^k m (\beta+1) \cdot \mathcal O(m\beta) = 2^k \mathcal O(m^2\beta^2)$.

    The total computation time of Step~2 across all $A\subseteq \ac$ with $|A|\ge 2$, $u\in V$, and $b'\in [\beta]_0$ requires a computation for every sub-subset.
    The time cost can be bounded by 

    \begin{align*}
        m(\beta+1)\sum_{j = 2}^k \binom{k}{j} 2^j \mathcal O(\beta) \le \mathcal O(m\beta^2) \sum_{j = 0}^k \binom{k}{j}2^j = 3^k\mathcal O(m\beta^2) \text.
    \end{align*}

    There, we use the Binomial Theorem for the last equation that states that for every $x,y\in \mathbb R$ and $k\in \mathbb N$, $(x+y)^k = \sum_{j = 0}^k \binom{k}{j}x^jy^{k-j}$.
    We apply it for $x = 2$ and $y = 1$.

    The total computation time of Step~3 can be derived as for Step~1 and bounded by $2^k\mathcal O(m^2\beta^2)$.

    Finally, obtaining the final minima requires one computation for each table entry, i.e., at most time $k\mathcal O(m\beta)$.

    Summing up everything, we obtain a total running time for filling the table of 

    \begin{align*}
        k \mathcal O(m\beta)+ 3^k\mathcal O(m\beta^2) + 2\cdot 2^k\mathcal O(m^2\beta^2) + k\mathcal O(m\beta) = 3^k\mathcal O(m\beta^2) + 2^k\mathcal O(m^2\beta^2)\text.
    \end{align*}

    Adding the time for deriving the budget mappings, we obtain a total running time of

    \begin{equation*}
        \mathcal O(m^3 \beta) + 3^k\mathcal O(m\beta^2) + 2^k\mathcal O(m^2\beta^2) = O(m^3 \beta + 3^k m\beta^2 + 2^k m^2\beta^2)\text.
    \end{equation*}

    This completes the proof.
\proofqed\end{proof}

\section{Variable Number of Agents}

\GBPhardness* 

\begin{proof}
Membership in \np{} is straightforward because the cost and feasibility of a given candidate solution to \Graph{} can be verified in polynomial time, see Observation~\ref{obs:NPmemb}.

For \np-hardness, we utilize the \np-complete problem \SetCover{} \citep{garey1979computers}, which is formally defined as follows: an instance $\langle \mathcal{U}, \mathcal S, \rho \rangle$ is given by a collection $\mathcal S$ of subsets of a finite ground set $\mathcal U$ and a positive integer $\rho$. 
It is a Yes-instance if there exists $\mathcal S'\subseteq\mathcal S$ with $|\mathcal S'|\le\rho$ that covers $\mathcal U$.

\begin{figure}[tb!]
\centering
\includegraphics[width=.65\textwidth, page=6, trim=30 210 390 78, clip]{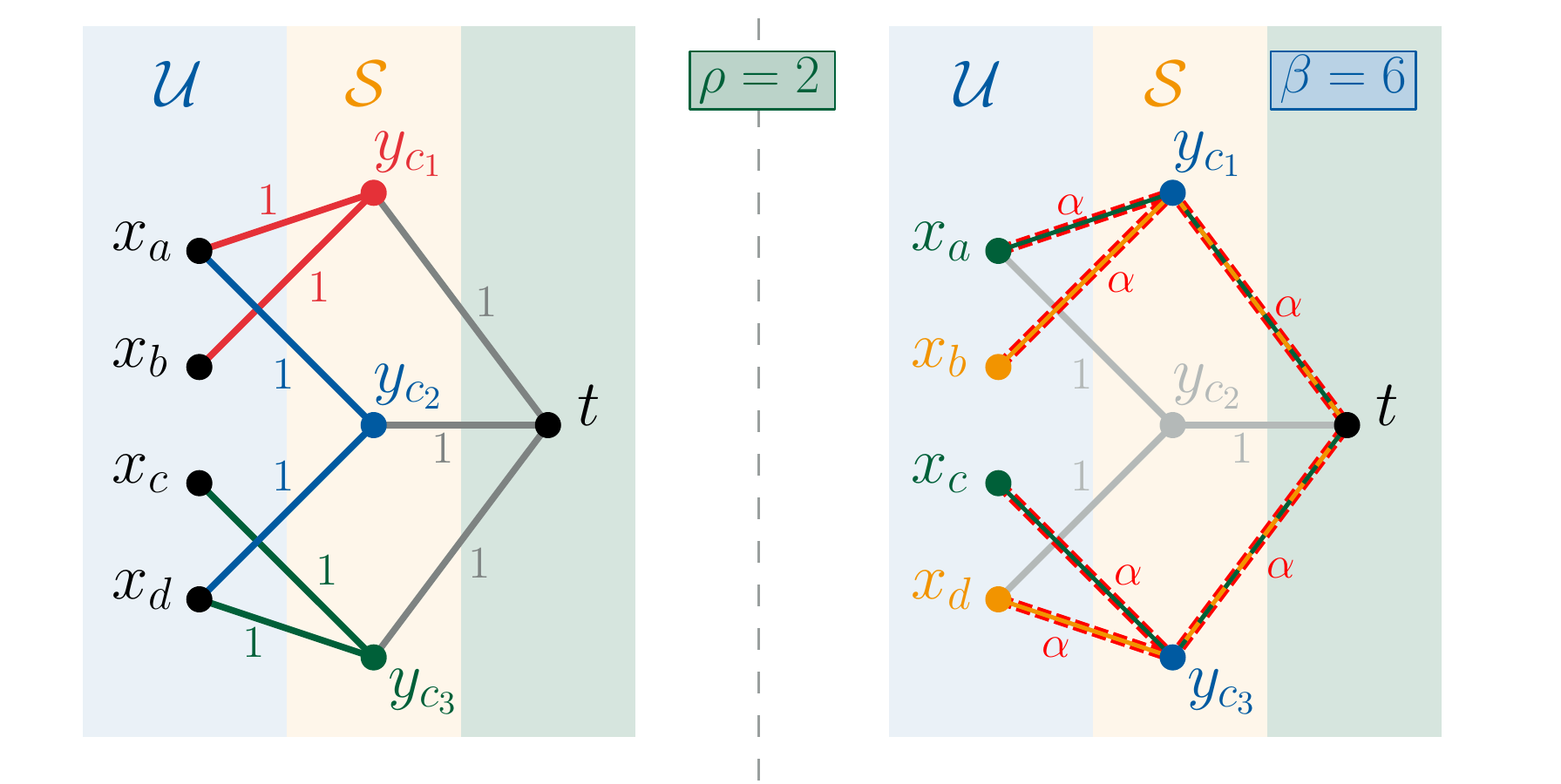}
\caption{Illustration of the reduction from \SetCover{} to \Graph{} for the instance $\langle \mathcal U, \mathcal S, 2\rangle$ with $\mathcal U = \{a,b,c,d\}$ and $\mathcal S = \{\{a,b\},\{b,c\},\{c,d\}\}$. The left panel shows the constructed \Graph{} instance and the right panel highlights a feasible solution corresponding to the set cover $\mathcal S' = \{\{a,b\},\{c,d\}\}$.
The dashed edges indicate the corresponding discounted edges.
\label{fig:optim:1}}
\end{figure}

Now, consider an instance $\ic = \langle \mathcal{U}, \mathcal S, \rho \rangle$ of \SetCover{} and a discount factor $\alpha \in [0,1)$.
Construct an instance $\bic' = \langle G, \mathcal{A}, \alpha, |\mathcal{U}| + \rho \rangle$ of \Graph{} with underlying graph $G = (V,E,w)$ as follows:
The vertices are $V = V_{\mathcal{U}} \cup V_{\mathcal S} \cup \{t\}$, $V_{\mathcal{U}} = \{x_u \colon u \in \mathcal{U}\}$ are called \emph{item vertices}, $V_{\mathcal S} = \{y_S \colon S \in \mathcal S\}$ \emph{subset vertices}, 
and we have a single \emph{target vertex} $t$.
The edges $E = E_I \cup E_T$ consist of \emph{item edges} $E_I = \{(x_u, y_S) \colon S \in \mathcal S, u \in S\}$ and \emph{target edges} $E_T = \{(y_S, t) \colon S \in \mathcal S\}$.
All edges have the same weight, i.e., we set $w(e) = 1$ for all $e \in E$.
The agent set is given by $\mathcal{A} = \multiset{(x_u, t) \colon u \in \mathcal{U}}$. 
An illustration of $\bic'$ is provided in \Cref{fig:optim:1}.
Let $\kappa\egsup = 2\alpha$ and $\kappa\utsup = 2\alpha |\mathcal{U}|$. 
Our reduced instance is the instance $\ic' = \langle \bic', \kappa^{\star} \rangle$ of \GraphDec{}.

Informally, this construction creates the bipartite incidence graph of items and subsets. 
Each item~$u\in \mathcal U$ is represented by an item vertex $x_u$ and each subset $S\in \mathcal S$ by a subset vertex $y_S$. 
We connect $x_u$ to $y_S$ whenever $u \in S$, and each $y_S$ links to the shared target $t$. 
Agent $(x_u, t)$ reaches $t$ by a two-step path from $x_u$ to $y_S$ and then to $t$, where $S$ must contain $u$.
The budget $\beta$ is set to 
$|\mathcal{U}| + \rho$, which allows to reduce the cost for a connection between each item vertex to some set vertex, and an additional $\rho$ edges from the sets to the target vertex. 
Intuitively, to reach our cost target, these connections have to correspond to a cover of $\mathcal U$ of size at most $\rho$.

We proceed with a formal proof of the correctness of our reduction.
We claim that $\mathcal{I}$ is a Yes-instance of \SetCover{} if and only if $\langle \bic', \kappa^{\star} \rangle$ is a Yes-instance of \GraphStar{}.

Since the minimal path cost per agent is $2\alpha$, we have, for every solution $F\subset E$
$c\egsup(F) = 2\alpha$ if and only if $c\utsup(F) = 2\alpha\,|\mathcal{U}|$.
Hence,
the result for \GraphSum{} follows immediately from the one for \GraphMax{}.

\noindent\textbf{($\Leftarrow$)}
Assume there is a feasible solution for $\bic'$ achieving egalitarian cost $\kappa\egsup$. 
Take any agent $(x_u,t)$ and let $P$ be a (shortest) $x_u$-$t$ path of cost $2\alpha$. 
If multiple paths tie, pick the one with fewest edges. Since $1 + \alpha > 2\alpha$, all edges on $P$ must be reduced, and if $\alpha > 0$, $P$ has exactly two edges. 
If $\alpha = 0$, one can argue by a finite sequence of replacements---see \Cref{fig:optim:8}---that any longer path can be shortened to two edges without increasing cost.

Since all paths are of length two, no paths of any pair of agents can share the same item edge, so at most $\rho$ reduced subset edges remain. Those $\rho$ edges correspond to a collection of at most $\rho$ subsets, and since each agent uses one of these subset edges, every item is covered. Thus they correspond a set cover of size at most $\rho$.

\begin{figure}[tb!]
\centering
\includegraphics[width=.65\textwidth, page=10, trim=30 180 390 135, clip]{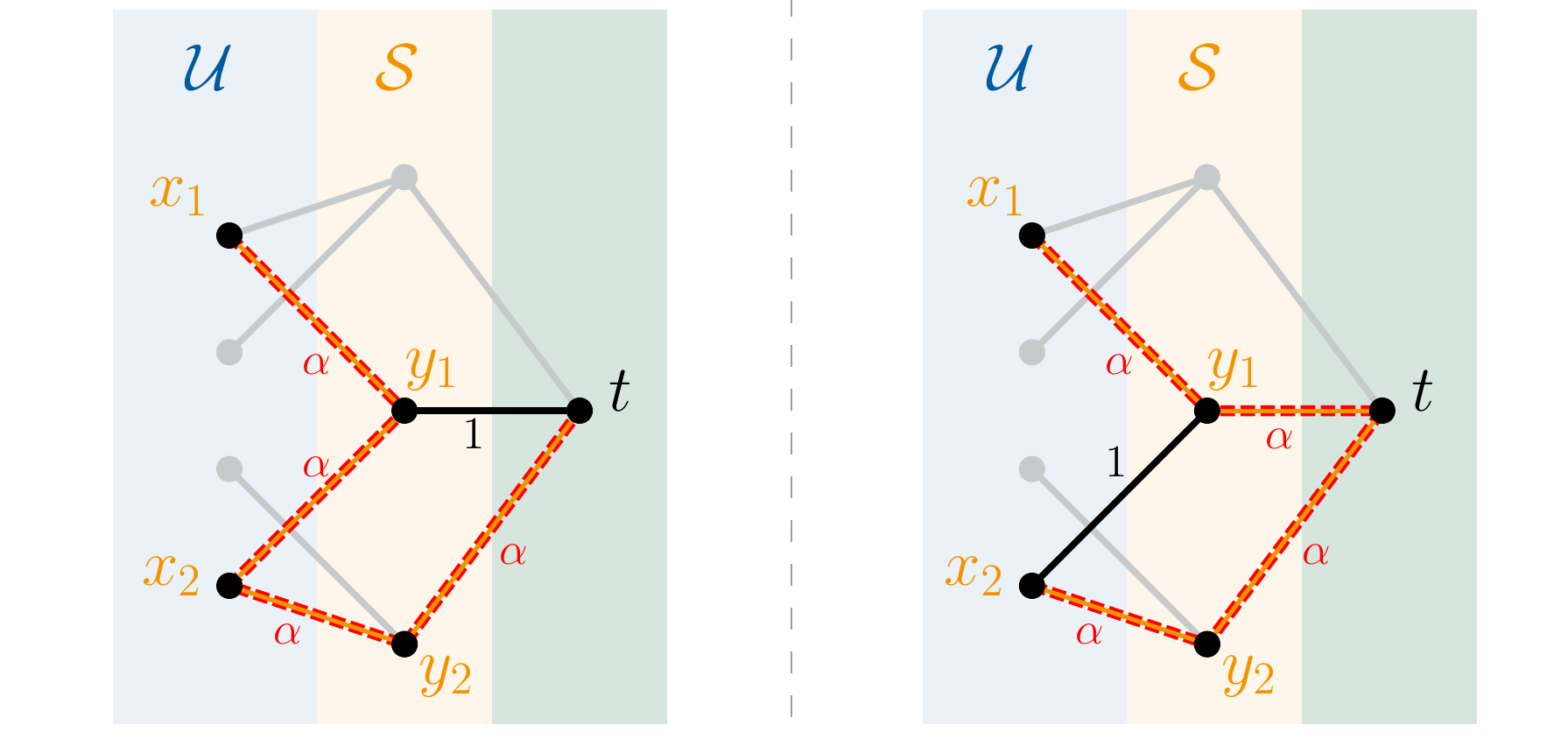}
\caption{
If the solution in the left panel constitutes a feasible solution, then the solution in the right panel is likewise feasible. All shortest paths passing through $y_1$ or $x_2$ can reach $t$ with the same or even lower costs.
\label{fig:optim:8}}
\end{figure}

\noindent\textbf{($\Rightarrow$)}
Conversely, let $\{S_1, \dots, S_\rho'\}$ be a cover of $\mathcal U$ of size $\rho'\le \rho$. 
Choose for each item vertex $x_u$ one reduced edge $\{x_u, y_{S_i}\}$ where $u \in S_i$, and include all $\rho'$ edges $\{y_{C_i}, t\}$. 
This uses at most $|\mathcal U| + \rho$ edges, and enables each agent to reach $t$ in two steps via reduced edges.
Hence, it leads to an egalitarian cost of $2\alpha$ for the \Graph{} instance.
\proofqed\end{proof}

\GBPInap*
\label{pro:GBPInap}
\begin{proof}
We claim that there exists no polynomial-time algorithm that achieves a $\gamma$-approximation for \GraphStar{} for any $\gamma > 1$.  
Suppose, for contradiction, that such an algorithm exists.  
Consider an instance $\mathcal{I}$ of \SetCover{} and its corresponding instance $\mathcal{I}'$ of \Graph{} constructed as described in the proof of Theorem~\ref{cor:graph0}, with $\alpha = 0$.  
From the previous proof, we know that $\mathcal{I}$ is a Yes-instance if and only if there exists a solution to $\mathcal{I}'$ with $c^{\star} = 0$.  
Now, if the approximation algorithm is applied to $\mathcal{I}'$, it must return a solution with cost $0 \leq c\egsup \leq \gamma \cdot \kappa_{opt}^{\star} = 0$ if $\mathcal{I}$ is a Yes-instance.  
Otherwise, when $\mathcal{I}$ is a No-instance, we have $c\egsup \geq \kappa_{opt}^{\star} > 0$.  
Thus, the algorithm distinguishes between Yes- and No-instances of \SetCover{} in polynomial time, which is impossible unless  $\p=\np$.

\proofqed\end{proof}

\GBPApproxBound*
\label{pro:GBPApproxBound}
\begin{proof}
Fix a discount factor $\alpha > 0$.  
We claim that no polynomial-time algorithm can achieve a $\gamma$-approximation for \GraphMax{} with $\gamma < \frac{1 + \alpha}{2\alpha}$.  
Assume, for contradiction, that such an algorithm exists.  
Let $\mathcal{I}$ be an instance of \SetCover{}, and construct the corresponding instance $\mathcal{I}'$ of \Graph{} as described  in the proof of Theorem~\ref{cor:graph0}.  
We established that if $\mathcal{I}$ is a Yes-instance, then $\mathcal{I}'$ has a solution with $c\egsup = 2\alpha$, and if $\mathcal{I}$ is a No-instance, then every feasible solution has $c\egsup > 2\alpha$.

We now argue that, in the No-instance case, any egalitarian-optimal solution must have $c\egsup \geq 1 + \alpha$.  
Suppose, for contradiction, that $2\alpha < c\egsup < 1 + \alpha$.  
This would imply that all agents still traverse only reduced edges, but at least one of them uses more than two.  
By the same reasoning used earlier (see Figure~\ref{fig:optim:8}), we can modify the solution by replacing item edges with subset edges so that each agent ultimately uses only two reduced edges.  
This transformation would yield a solution contradicting the assumption that $\mathcal{I}$ is a No-instance.

Therefore, the approximation algorithm would return a solution with $c\egsup \leq \gamma \cdot \kappa_{opt}\egsup < 1 + \alpha$ when $\mathcal{I}$ is a Yes-instance, and $c\egsup \geq \kappa_{opt}\egsup \geq 1 + \alpha$ when $\mathcal{I}$ is a No-instance.  
Thus, the algorithm would be able to decide \SetCover{} in polynomial time, which is not possible unless $\p=\np$.

\proofqed\end{proof}

\section{Comparison with \RWP{}}\label{app:railway} 
The \rwp{} (\Rail), introduced by \addauthor{He et al.~}\citet{HBL+24}, models a transportation network as a graph where vertices represent cities, edges denote connections between them, and edge weights correspond to distances. Each city has a specified demand to travel to other cities, and the goal is to design a railway network connecting these locations. In their framework, the cost of using a railway edge equals its distance, while non-rail edges incur a significantly higher cost, scaled by a factor $\zeta$.\footnote{In the original work, this parameter is denoted by $K$.} 
This is conceptually analogous to commuters having to take a bus (as opposed to base cost in the \Graph{} model).

A key question is whether representing agents via a demand function leads to improved runtime efficiency compared to modeling agents as a multi-set of terminal pairs. The demand-based representation proves particularly advantageous when the number of agents $n$ far exceeds the number of cities $m$, or more precisely, the number of possible city pairs $m^2$. In the \Graph{} setting, agents can be duplicated arbitrarily, leading to a potentially large multi-set. This issue can be alleviated by adopting a demand matrix. Rather than storing each agent individually, we use a two-dimensional matrix over city pairs to count agent occurrences. This is effectively equivalent to using a demand function and becomes efficient especially when $n > m^2$.

Adapting our notation an instance $\bic = \langle V, E, w, \tau, \beta\rangle$ of \Rail{}, consists of a set $V$ of $m$ vertices, a set $E$ of $l$ edges, weight function $w\colon E\to\mathbb{Q}_{\ge0}$, a symmetric demand function $\tau: V^2 \rightarrow \Nz$ with $\tau(u,v) = \tau(v,u)$ and $\tau(u,u) =0$, and $\beta\in\mathbb{N}$ as budget. A solution $R\subseteq E$ is \emph{feasible} if $w(R)=\sum_{e\in R}w(e)\le\beta$. The objectives \RailMax{} and \RailSum{} are defined analogously, with a target $\kappa \in \mathbb{Q} \cup \{\infty\}$, where any finite cost satisfies $\kappa = \infty$ but not infinite cost.

To illustrate the parallels between the two frameworks, consider an example similar to a previously discussed scenario. As shown in Figure~\ref{fig:optim:4}, assume two agents $(s,t)$ and $(s',t')$ are present. With a discount factor of $\alpha = \frac{1}{2}$ and a budget of one, reducing a single edge with weight ten to a weight of five results in a total travel cost of $19$.
This can be translated into a railway instance by keeping the graph structure the same but using a penalty factor $\zeta = \alpha^{-1} = 2$. Instead of listing individual agents, we define a demand function $\tau$ with $\tau(s,t) = \tau(s',t') = 1$ and all other entries zero\footnote{Symmetry of the demand function is omitted for simplicity but considered in the model.}. The objective is now to determine which edge to preserve at its original cost, while all others are scaled by $\zeta$. Clearly, if the same edge is retained, the resulting total cost is exactly twice the previous value, that is, $38$, consistent with the scaling factor.

\begin{figure}[tb!]
\centering
\includegraphics[width=.9\textwidth, page=5, trim=30 320 540 70, clip]{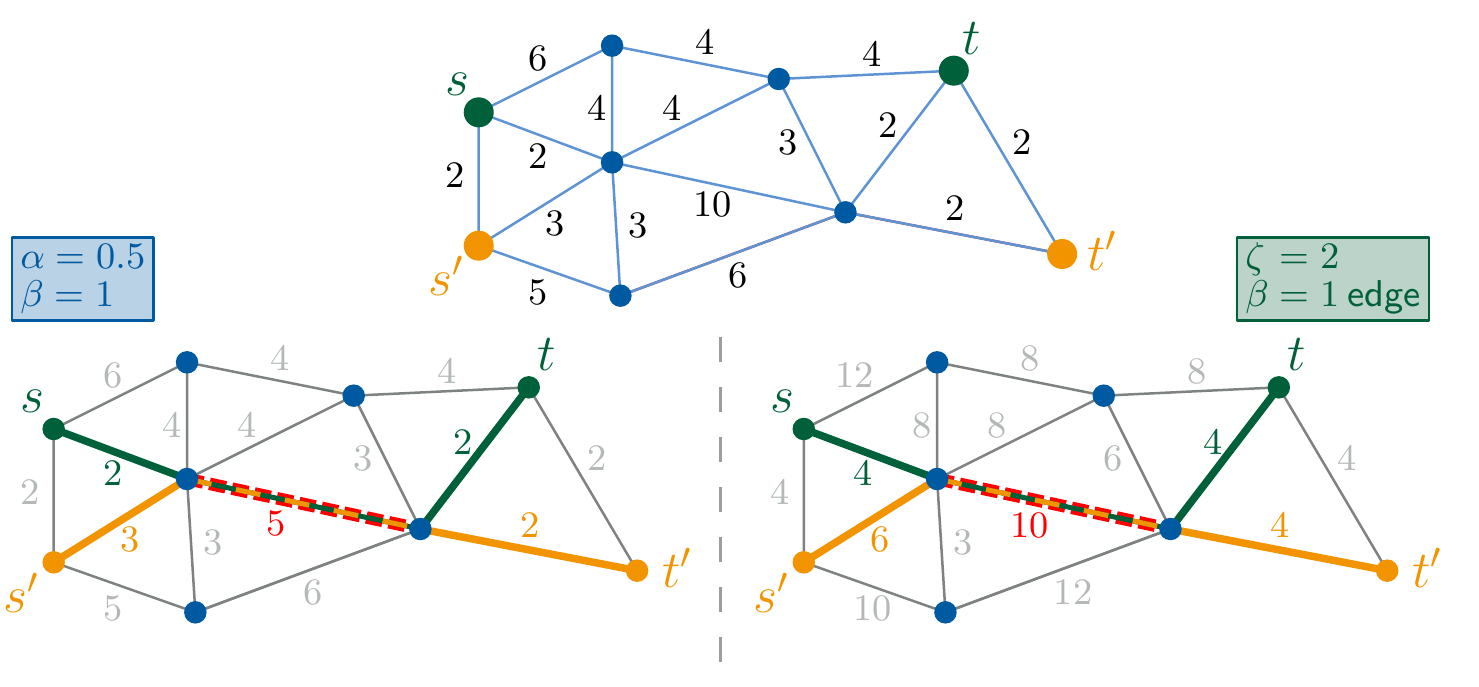}
\caption{
Comparison of \Rail{} and \Graph{}: the top diagram shows the base graph; the left subfigure highlights a solution minimizing utilitarian cost in \Graph{} (dashed edges), while the right subfigure displays the corresponding solution in \Rail{}, where all edge weights are scaled by $\zeta = 2$. The budget allows modifying exactly one edge in each case.
\label{fig:optim:4}}
\end{figure}

In contrast to the \Graph{} model, the edge cost function in \Rail{} given a solution $R$ is defined as $w_R(e) = w(e)$ if $e \in R$ and $\zeta \cdot w(e)$ otherwise. The cost metrics are given by $c\egsup(R) := \max_{(u,v) \in V^2} \tau(u,v) \cdot \pi_R(u,v)$ and $c\utsup(R) := \sum_{(u,v) \in V^2} \frac{1}{2} \tau(u,v) \cdot \pi_R(u,v)$, where the factor $\frac{1}{2}$ compensates for symmetric counting. While \addauthor{He et al.~}\citet{HBL+24} also study combined utilitarian and egalitarian objectives, we focus only on the pure forms here.

At first glance, the two models appear closely related. In fact, transforming instances between them is straightforward: cities $u$ and $v$ with demand $\tau(u,v) = d$ can be modeled in \Graph{} by creating $d$ agents with terminals $(u,v)$, and vice versa. The penalty factor $\zeta$ and discount factor $\alpha$ are reciprocals, i.e., $\zeta = \alpha^{-1}$, and target costs must be scaled accordingly by either $\alpha$ or $\zeta$. We use the shorthand $0^{-1} := \infty$ throughout.

Despite their structural similarities, a closer comparison of \Graph{} and \Rail{} reveals key differences:

\underline{Cost interpretation:}
In \Graph{}, edge weights reflect base cost, and selected edges are discounted by $\alpha$. In \Rail{}, purchased edges retain their original weights, while all others are penalized by a factor $\zeta \in \mathbb{Q}_{>1} \cup \{\infty\}$. When $\zeta$ is finite, this distinction has no effect on solutions. However, when $\zeta = \infty$ (i.e., $\alpha = 0$), feasibility diverges: if the budget does not allow all required travel via selected edges, the cost in \Rail{} becomes infinite. In contrast, \Graph{} still permits meaningful cost minimization. Consider two instances $\langle \ic, a \rangle$ and $\langle \ic, b \rangle$ of \GraphStar{} with $0 < a < b$. It is possible that the first is a No-instance and the second a Yes-instance, but both map to the same \RailStar{} instance $\langle \ic', \infty \rangle$.

\underline{Egalitarian fairness:}
In \Graph{}, egalitarian cost is evaluated per agent, while in \Rail{} it is computed over city pairs and weighted by demand. This difference can lead to divergent solutions. The distinction disappears when the demand function is binary, corresponding to a unique agent set.

\underline{Budget consumption:}
The models diverge most in how they treat budget. In \Graph{}, the budget limits the number of reduced edges. In \Rail{}, the budget accounts for the total weight of selected edges. As a result, \Rail{} solutions often favor reducing shorter edges. This makes sense in railway design, where construction cost grows with distance, but is less appropriate for bus networks, where the budget reflects the number of routes, not their length.

This discrepancy can be resolved by assuming binary edge weights: zero-weight edges require no budget, and all others have the same cost. While this is a strong simplification, under this assumption the models become directly convertible.

\end{document}